\documentclass[%
superscriptaddress,
 amsmath,amssymb,
 aps,
 physrev,
twocolumn,
]{revtex4-2}

\usepackage{graphicx}
\usepackage{dcolumn}
\usepackage{bm}
\usepackage{hyperref}

\usepackage{xcolor}
\usepackage{physics}

\usepackage{tikz}
\usetikzlibrary{quantikz2}
\usetikzlibrary{angles, quotes}

\usepackage{amsthm}
\usepackage{enumitem}

\usepackage{algorithm}
\usepackage{algpseudocode}

\newtheorem{theorem}{Theorem}
\newtheorem{corollary}{Corollary}
\newtheorem{lemma}{Lemma}
\newtheorem{proposition}{Proposition}

\theoremstyle{definition}
\newtheorem{definition}{Definition}

\newcommand{\ceil}[1]{\lceil{#1}\rceil}

\begin{document}

\preprint{APS/123-QED}

\title{\textbf{Efficient benchmarking of logical magic state} 
}%

\author{Su-un Lee}%
\email{suun@uchicago.edu}
\thanks{These authors contributed equally}
\affiliation{%
 Pritzker School of Molecular Engineering, The University of Chicago, Chicago, IL 60637, USA
}%

\author{Ming Yuan}%
\email{yuanming@uchicago.edu}
\thanks{These authors contributed equally}
\affiliation{%
 Pritzker School of Molecular Engineering, The University of Chicago, Chicago, IL 60637, USA
}%

\author{Senrui Chen}%
\affiliation{%
 Pritzker School of Molecular Engineering, The University of Chicago, Chicago, IL 60637, USA
}%

\author{Kento Tsubouchi}%
\affiliation{Department of Applied Physics, University of \mbox{Tokyo, 7-3-1} Hongo, Bunkyo-ku, Tokyo 113-8656, Japan}

\author{Liang Jiang}%
\email{liang.jiang@uchicago.edu}
\affiliation{%
 Pritzker School of Molecular Engineering, The University of Chicago, Chicago, IL 60637, USA
}%

\date{May 16, 2025}

\begin{abstract}
High-fidelity logical magic states are a critical resource for fault-tolerant quantum computation, enabling non-Clifford logical operations through state injection. However, benchmarking these states presents significant challenges: one must estimate the infidelity $\epsilon$ with multiplicative precision, while many quantum error-correcting codes only permit Clifford operations to be implemented fault-tolerantly. Consequently, conventional state tomography requires $\sim1/\epsilon^2$ samples, making benchmarking impractical for high-fidelity states. In this work, we show that any benchmarking scheme measuring one copy of the magic state per round necessarily requires $\Omega(1/\epsilon^2)$ samples for single-qubit magic states. We then propose two approaches to overcome this limitation: (i) Bell measurements on two copies of the twirled state and (ii) single-copy schemes leveraging twirled multi-qubit magic states. Both benchmarking schemes utilize measurements with stabilizer states orthogonal to the ideal magic state and we show that $O(1/\epsilon)$ sample complexity is achieved, which we prove to be optimal. Finally, we demonstrate the robustness of our protocols through numerical simulations under realistic noise models, confirming that their advantage persists even at moderate error rates currently achievable in state-of-the-art experiments.
\end{abstract}

\maketitle

\section{Introduction}

Fault-tolerant quantum computation (FTQC) relies on quantum error correction (QEC) to encode quantum information in a larger Hilbert space, thereby suppressing logical errors. Many widely adopted QEC codes support fault-tolerant Clifford gates and Pauli measurements through transversal operations~\cite{Steane_1996, Steane1996Error, gottesmanStabilizerCodesQuantum1997, gottesmanFaultTolerantQuantumComputation1999} or lattice surgery~\cite{horsmanSurfaceCodeQuantum2012, fowler2019lowoverheadquantumcomputation}. However, the Eastin-Knill theorem restricts fault-tolerant implementations of non-Clifford gates~\cite{eastinRestrictionsTransversalEncoded2009}, which are essential for universal quantum computation~\cite{gottesman1998heisenbergrepresentationquantumcomputers, aaronsonImprovedSimulationStabilizer2004a, Kitaev_1997}. To overcome this limitation, magic states—non-stabilizer states enabling non-Clifford gates via state injection—have emerged as a critical resource for FTQC~\cite{bravyiUniversalQuantumComputation2005}.

While directly preparing high-fidelity magic states within QEC codes is challenging, magic state distillation provides a practical solution by consuming multiple noisy copies to produce higher-fidelity ones~\cite{bravyiUniversalQuantumComputation2005, Litinski2019magicstate, Gidney2019efficientmagicstate}. Distilling a \textit{logical} magic state is particularly important for implementing high-fidelity non-Clifford logical gates. With this importance, magic state distillation on [[7, 1, 3]] and [[17, 1, 5]] color codes~\cite{Steane_1996,Steane1996Error,bombinTopologicalQuantumDistillation2006} has recently been demonstrated using reconfigurable neutral atom arrays~\cite{rodriguezExperimentalDemonstrationLogical2024}.

Given the central role of these logical magic states in FTQC, their accurate characterization is crucial. However, benchmarking high-fidelity logical magic states faces significant challenges. First, while the logical infidelity $\epsilon$ can be very small, one needs to estimate it with multiplicative precision, as the estimation error must be at least in the same order as $\epsilon$ itself. Second, one only has a limited set of quantum operations consisting of fault-tolerant Clifford gates and Pauli measurements, for logical errors from the benchmarking circuit being negligible to $\epsilon$. As a result, conventional methods that perform single-qubit Pauli measurements require $\sim 1/\epsilon^2$ samples, and this led to a large sampling overhead for recent experimental demonstrations of logical magic states~\cite{YeLogicalMagicState2023, GuptaEncodingBreak-even2024, mayer2024benchmarkinglogicalthreequbitquantum, Wang2024_832colorcode, rodriguezExperimentalDemonstrationLogical2024}. Especially, Ref.~\cite{rodriguezExperimentalDemonstrationLogical2024} reports fidelity of $0.994$ at the logical level, consuming $\approx 30{,}000$ magic states, which will become impractical as the infidelity is improved.

In this work, we present methods to circumvent the large $\Omega(1/\epsilon^2)$ sample complexity and achieve $O(1/\epsilon)$ scaling for magic state benchmarking. Specifically, we first establish a fundamental limitation for \textit{single-copy benchmarking schemes}, which measure each copy of the magic state individually, showing that any single-qubit magic state requires $\Omega(1/\epsilon^2)$ samples to benchmark. We then propose two strategies to overcome this limitation: (i) perform Bell measurements on two copies of the twirled magic state, and (ii) perform single-copy benchmarking schemes on twirled multi-qubit magic states. Both approaches achieve $O(1/\epsilon)$ sample complexity, which we prove is optimal. For both protocols, identifying stabilizer states orthogonal to the ideal magic state plays a central role in improving sampling overhead.

Finally, we evaluate the robustness of our proposed protocols under a range of experimental imperfections, including gate, state preparation, and measurement errors. Our numerical simulations confirm that—even under this realistic setup—our protocols retain their advantage, substantially reducing the sampling overhead compared to standard state tomography. Specifically, at a moderate noise level that is already achievable by modern experiments, our method reduces the sampling overhead by over two orders of magnitude, suggesting practical feasibility for near-term experiments.

\section{Problem setup}

The goal of magic state benchmarking is to estimate the infidelity of an ideal magic state and a noisy magic state. Specifically, given a target $n$-qubit magic state $\ket{\psi}$ and an arbitrary density matrix $\rho$ that we prepare, the fidelity is given by $F = \langle\psi|\rho|\psi\rangle$. Here, $\rho$ can be either physical or logical state. If it is a logical state, we assume that $\rho$ is the state after an error correction or error detection circuit has been applied, so that $\rho$ lies within the logical space. We consider the benchmarking task of estimating the infidelity $\epsilon = 1-F$ with a multiplicative error, $|\hat{\epsilon}- \epsilon| \le r \epsilon$, for some $0<r<1$. Multiplicative precision is essential in this context: additive precision alone becomes inadequate for small values of infidelity, as trivially reporting $\hat{\epsilon}=0$ would already achieve a small additive error, rendering benchmarking ineffective.

Motivated by the fact that many quantum error-correcting codes support only Clifford gates and Pauli measurements fault-tolerantly, we constrain the benchmarking schemes to use only Clifford gates and Pauli measurements. We first define a \textit{single-copy benchmarking scheme}. This scheme consists of $N$-rounds of measurement, and for $i$-th measurement, we prepare $\rho$ and apply a Clifford gate $U_i$. We then measure every qubit in the computational basis and get an outcome $x \in \{0,1\}^{n}$. Here, we allow the single-copy scheme to be adaptive, meaning the gate $U_i$ can be chosen based on the previous measurement outcomes. One may consider more general schemes involving ancillary qubits, but such ancillas do not affect the sample complexity (see Appendix~\ref{sec:ancilla}). The single-copy scheme includes a wide range of benchmarking schemes such as standard state tomography and direct fidelity estimation~\cite{flammiaDirectFidelityEstimation2011, da2011practical}.

As we restrict our benchmarking scheme to consist of Clifford gates, it is useful to introduce the stabilizer states, which are generated by applying Clifford circuits to computational basis states. Specifically, we denote the Clifford groups associated with $n$-qubits as $\mathcal{C}_n$ and define the set of stabilizer states as
\begin{equation}
    {\rm Stab}_n = \{U\ket{0^n}\in\mathcal{H} :U \in \mathcal{C}_{n}\}.
\end{equation}
It is immediately seen that the output probability $P(x)$ of measuring $x\in \{0,1\}^n$ for $i$-th round of measurement is the overlap between the input state $\rho$ and a stabilizer state $U_i^{\dagger}\ket{x}$, i.e., $P_i(x)= \langle x|U_i \rho U_i^{\dagger}| x\rangle$.

\section{Benchmarking single-qubit magic state}

\begin{figure}
    \centering
    \includegraphics[width=\columnwidth]{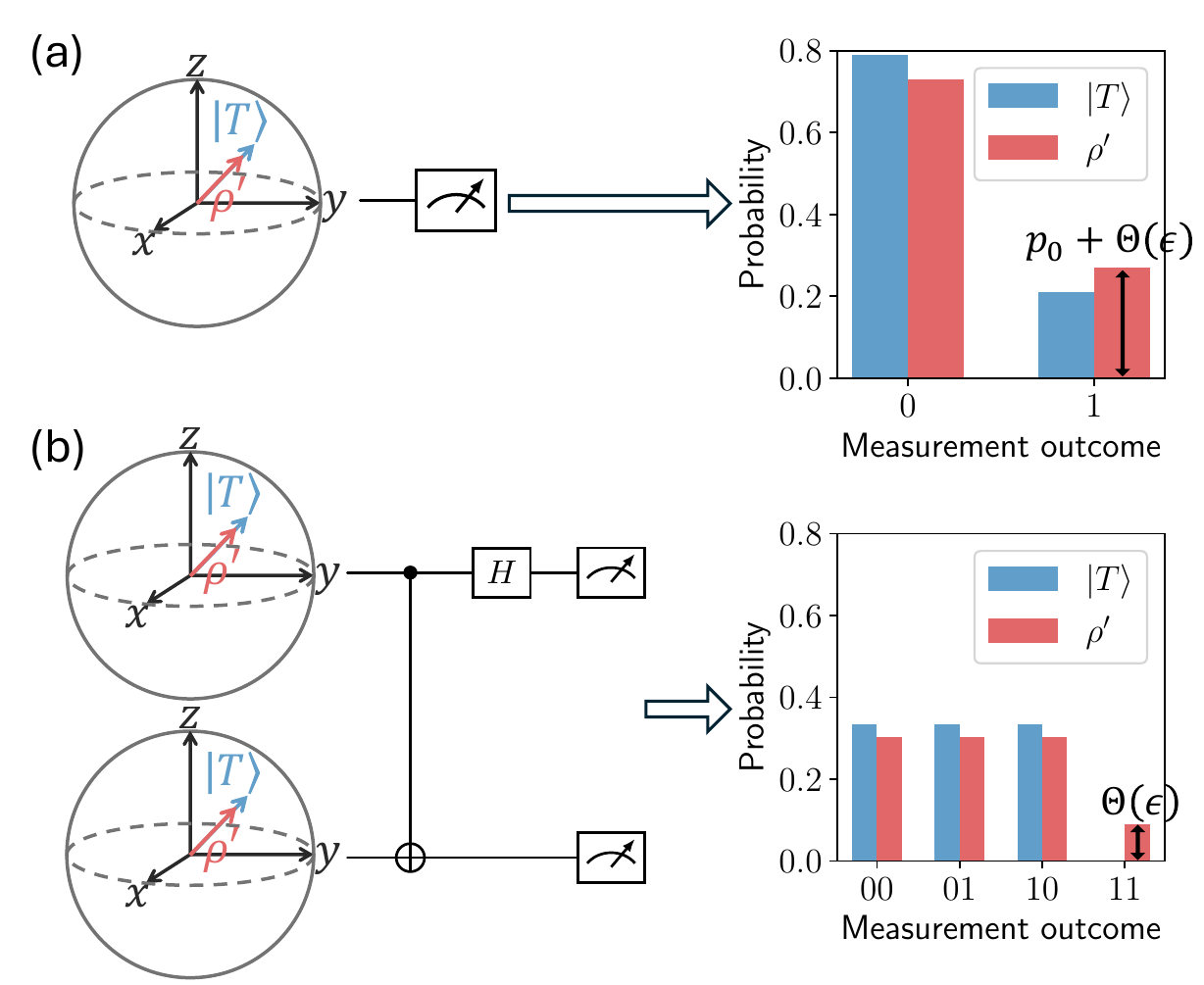}
    \caption{(a) Single-copy benchmarking scheme for $\ket{T}$ state with corresponding measurement output distribution, and (b) Bell measurement benchmarking scheme for $\ket{T}$ state with corresponding measurement output distribution.}
    \label{fig:fig1}
\end{figure}

We illustrate the limitation of the single-copy benchmarking scheme and how a joint measurement can overcome this limitation with an explicit example. In particular, we highlight the role of stabilizer states that are orthogonal to the ideal magic state.

Let us consider a single-copy benchmarking scheme for the following single-qubit magic state,
\begin{equation}
    \ketbra{T}{T} = \frac{1}{2}\left(I +\frac{1}{\sqrt{3}}(X+Y+Z)\right),
\end{equation}
which enables non-Clifford $\pi/6$-phase rotations and serves as a fundamental resource for FTQC~\cite{bravyiUniversalQuantumComputation2005}. Suppose we prepare a single-qubit state $\rho$ that has infidelity $\epsilon$ with respect to $\ket{T}$. We can convert the state into $\rho'$ with the same infidelity,
\begin{align}
    \rho'&=(1-\epsilon)\ketbra{T}{T}+\epsilon\ketbra{T^\perp}{T^\perp}\\
    &=\frac{1}{2}\left(I+\frac{1-2\epsilon}{\sqrt{3}}(X+Y+Z)\right)
\end{align}
by randomly applying some Clifford gates~\cite{bravyiUniversalQuantumComputation2005}, which is called state twirling. Here, $\ket{T^\perp}$ denotes the unique state orthogonal to $\ket{T}$.

Because of the symmetry $\rho'$ has, a measurement on $\rho$ proceeded by any single-qubit Clifford gate gives the same statistics. Therefore, without loss of generality, consider the case where we measure $\rho'$ in the computational basis $N$ times. It is straightforward that for each measurement, the outcome is sampled from the Bernoulli distribution of ${\rm Bern}(p_0+\frac{\epsilon}{\sqrt{3}})$, where $p_0 = \frac{1}{2}(1-\frac{1}{\sqrt{3}})$ [Fig.~\ref{fig:fig1}(a)]. Then, with $N$ samples of $x^{(1)},\dots,x^{(N)} \sim {\rm Bern}(p_0+\frac{\epsilon}{\sqrt{3}})$, the most straightforward estimator is $\hat{\epsilon}=\sqrt{3}\left(\frac{1}{N}\sum_{i=1}^{N}x^{(i)} - p_0\right)$, which has the standard deviation of $\Delta\hat{\epsilon}=\sqrt{3(p_0+\frac{\epsilon}{\sqrt{3}})(1-p_0-\frac{\epsilon}{\sqrt{3}})/N} =\Theta(1/\sqrt{N})$. Therefore, achieving a constant multiplicative precision in $\epsilon$ requires $N = \Theta(1/\epsilon^2)$, making the single-copy benchmarking impractical for small $\epsilon$.

We remark that this $1/\epsilon^2$ scaling stems from the fact that all stabilizer states, e.g., $\ket{0}$ and $\ket{1}$, have non-zero overlap with the $\ket{T}$ state. This leads to that the measurement output should always be sampled from ${\rm Bern}(c + O(\epsilon))$ for some constant $0<c<1$, and it brings constant standard deviation $\sqrt{c(1-c)}$ of each sample even when the infidelity is zero. However, if a stabilizer state $\ket{s}$ existed such that $\bra{s}\ket{T}=0$, the distribution of the measurement with this state would instead be a Bernoulli distribution without constant offset, i.e., ${\rm Bern}(\Theta(\epsilon))$. Then, the standard deviation of each sample becomes $\Theta(\sqrt{\epsilon})$, and thus that of the estimator is $\Theta(\sqrt{\epsilon/N})$. It leads to $N=O(1/\epsilon)$ sample complexity. Therefore, the existence of a stabilizer state orthogonal to the ideal magic state determines whether efficient benchmarking is possible.

Now, we demonstrate how a joint measurement on multiple copies of $\rho$ reduces sampling overhead. Note that while no stabilizer state is orthogonal to the single-copy magic state $\ket{T}$, there exists a stabilizer state orthogonal to two copies of it, i.e., $\ket{T}^{\otimes 2}$.  Specifically, the following stabilizer state, $\frac{1}{\sqrt{2}}(\ket{01}-\ket{10})={\rm CNOT}_{12}(H\otimes I)\ket{11}$, constitutes a singlet state and thus antisymmetric under SWAP operation. Meanwhile, $\ket{T}^{\otimes{2}}$ is symmetric under SWAP, so $\frac{1}{\sqrt{2}}(\ket{01}-\ket{10})$ has zero overlap with $\ket{T}^{\otimes 2}$.

This motivates us to consider the \textit{Bell measurement scheme}. For an even number $N$, the Bell measurement scheme consists of $N/2$ rounds of measurement consuming $N$ copies of $\rho$ in total. In each round, two copies of $\rho$ are prepared, and we convert them to twirled states $\rho'$. We then apply ${\rm CNOT}_{1,2}(H\otimes I)$, followed by computational basis measurements on all qubits [Fig.~\ref{fig:fig1}(b)]. Therefore, we get an outcome $x \in \{0,1\}^2$ with probability of $\langle{\rm Bell}_{x}|\rho'^{\otimes 2}|{\rm Bell}_{x}\rangle$, where we denote $\ket{{\rm Bell}_{x}}={\rm CNOT}_{12}(H\otimes I)\ket{x}$. By the uniqueness of the sinlget state, we have $\ket{{\rm Bell_{11}}} \propto \frac{1}{\sqrt{2}}(\ket{TT^\perp}-\ket{T^\perp T})$. Therefore, $\langle{\rm Bell}_{11}|\rho'^{\otimes 2}|{\rm Bell}_{11}\rangle=\epsilon(1-\epsilon)=\Theta(\epsilon)$. This enables sampling from ${\rm Bern}(\Theta(\epsilon))$, and thus we can achieve $N=O(1/\epsilon)$ sample complexity for estimating $\epsilon$ to constant multiplicative precision.

It is noteworthy that the Bell measurement scheme does not require two copies of input magic states to be identical. When the two copies are $\rho_1$ and $\rho_2$ with infidelities $\epsilon_1$ and $\epsilon_2$, respectively, the Bell measurement scheme estimates the average infidelity $(\epsilon_1+\epsilon_2)/2$. Also, we do not need to apply twirling operations to the both copies; applying it to one copy suffices to achieve the Bell measurement scheme (see Appendix~\ref{sec:swap_test}).

We remark that, beyond this specific application, the Bell measurement is broadly utilized across quantum resource theories regarding magic states: it enables efficient certification and quantification of the nonstabilizerness of quantum states~\cite{haung2023_ScalableMagic, haug2025efficientwitnessingtestingmagic}, and plays a central role in tailoring and characterizing resource states for measurement-based quantum computation~\cite{park2025efficientnoisetailoringdetection}.

\section{Main results}

Now, we formalize the limitations of single-copy benchmarking schemes and establish conditions under which benchmarking can be performed efficiently. The proofs of the main theorems are provided in Appendix~\ref{sec:proofs}.

We first establish a fundamental constraint on single-copy benchmarking, showing that for certain magic states, any benchmarking scheme requires $\Omega(1/\epsilon^2)$ samples to achieve multiplicative precision.

\begin{theorem}
    Let $\ket{\psi}$ be an $n$-qubit state having no orthogonal stabilizer state. Consider any single-copy benchmarking scheme producing an estimator $\hat{\epsilon}$ of the infidelity $\epsilon = 1 - \langle\psi|\rho|\psi\rangle$, satisfying the multiplicative error $|\hat{\epsilon}-\epsilon| \leq r\epsilon$ with probability at least $2/3$, for an arbitrary input state $\rho$. Such a protocol necessarily requires $N = \Omega\left(1/r^2\epsilon^2\right)$ copies of $\rho$.
    \label{thm:no-go}
\end{theorem}

Notably, no single-qubit non-stabilizer state $\ket{\psi}$ has an orthogonal stabilizer state because if such a stabilizer state existed, $\ket{\psi}$ must be a stabilizer state itself. Consequently, Theorem~\ref{thm:no-go} implies that any single-copy benchmarking scheme for a single-qubit magic state always suffers from $\Omega(1/\epsilon^2)$ sample complexity.

To overcome the constraints imposed by Theorem~\ref{thm:no-go}, we explore two approaches: (1) performing a joint measurement on multiple copies of $\rho$, and (2) considering multi-qubit magic states. A key tool in both approaches is state twirling, which simplifies the structure of noisy magic states while preserving their fidelity.

\begin{definition}
    Given an $n$-qubit magic state $\ket{\psi}$, we define its \emph{twirling group} as $G_{\psi} = \{U \in \mathcal{C}_n : U \ket{\psi} \propto \ket{\psi}\}$.
    The \emph{twirling operation} with respect to $\ket{\psi}$ is defined as the application of $U$ uniformly sampled from $G_{\psi}$, leading to the transformation
    \begin{equation}
        \rho \mapsto \rho'=\frac{1}{|G_{\psi}|}\sum_{U \in G_{\psi}} U\rho U^{\dagger}.
    \end{equation}
\end{definition}

Note that $G_{\psi}$ forms a subgroup of $\mathcal{C}_n$ by its definition. Since $\langle \psi| \rho | \psi \rangle = \langle \psi|U \rho U^\dagger| \psi \rangle$ for all $U \in G_{\psi}$, fidelity is invariant under the twirling operation. By the definition of $G_{\psi}$, $\mathcal{H}_0 = {\rm span}\{\ket{\psi}\}$ is an irreducible representation of $G_{\psi}$. (See Appendix~\ref{sec:reprethry} for a self-contained review of representation theory.) We decompose the Hilbert space as $\mathcal{H}=\mathcal{H}_0 \oplus \mathcal{H}_1 \oplus\dots \oplus \mathcal{H}_k$, where $\mathcal{H}_1, \dots, \mathcal{H}_k$ are the other irreducible representations of $G_{\psi}$. Then, for $\rho' = \frac{1}{|G_\psi|}\sum_{U\in G_\psi}U\rho U^\dagger$, we decompose it as
\begin{equation}
    \rho' = \sum_{i,j=0}^k\rho'_{ij},
\end{equation}
where $\rho'_{ij} = \Pi_i \rho' \Pi_j$. Here, $\Pi_j$ is the projection onto $\mathcal{H}_j$ for $j=0,1,\dots,k$. Since $\mathcal{H}_j$'s are irreducible representation, $V\Pi_i = \Pi_i V$ for $i=0,1,\dots, k$ and all $V \in G_\psi$. Therefore, for all $V \in G_\psi$,
\begin{align}
    V\rho'_{ij}
    &= \frac{1}{|G_\psi|}\sum_{U\in G_\psi}\Pi_iVU\rho U^\dagger \Pi_j\\
    &= \frac{1}{|G_\psi|}\sum_{U\in G_\psi}\Pi_iU\rho (U)^\dagger V\Pi_j \\
    &= \rho'_{ij}V,
\end{align}
and thus each $\rho'_{ij}$ is a $G_\psi$-linear map from $\mathcal{H}_j$ to $\mathcal{H}_i$. Due to Schur's lemma, we have for all $i=0,1,\dots, k$, $\rho'_{ii} = \lambda_i\Pi_i$ for some $\lambda_i \in \mathbb{R}$, and $\rho'_{ij}=0$ if $\mathcal{H}_i \not\simeq \mathcal{H}_j$ which denotes that $\mathcal{H}_i$ and $\mathcal{H}_j$ are not equivalent representations of $G_\psi$.

State twirling is useful especially when $\mathcal{H}_0 \not\simeq \mathcal{H}_j$ for all $j=1,\dots, k$. In this case, the state after twirling becomes
\begin{equation}
    \rho'= (1-\epsilon)\Pi_0 \oplus \sum_{i,j=1}^k \rho'_{ij}.
\end{equation}
Note that since $\rho'$ is positive semi-definite, $\sum_{i,j=1}^k \rho'_{ij}$ also positive semi-definite. Therefore, for some density matrix $\sigma$ of the subspace $\mathcal{H}_1 \oplus \dots \oplus \mathcal{H}_k$,
\begin{equation}
    \rho' = (1-\epsilon)\ketbra{\psi}{\psi} \oplus \epsilon \sigma.
    \label{eq:twirled_state}
\end{equation}

With the assumption that $\mathcal{H}_0 \not\simeq \mathcal{H}_j$ for all $j = 1, \dots, k$, we now show how joint measurements on two copies of $\rho$ overcome the limitations established in Theorem~\ref{thm:no-go}. The Bell measurement scheme can be generalized to benchmark multi-qubit magic states: preparing two copies of $\rho'$, we apply $\prod_{i=1}^{n}{\rm CNOT}_{i,i+n}(H_{i}\otimes I_{i+n})$, followed by computational basis measurements on all qubits. Crucially, this circuit effectively realizes the SWAP test~\cite{buhrmanQuantumFingerprinting2001}, thereby estimating purity of $\rho'$. Specifically, denoting $P_{\rm odd}$ as the probability of obtaining a measurement outcome $x \in \{0,1\}^{2n}$ such that $\sum_{i=1}^n x_i \cdot x_{i+n}$ is an odd number, we have
\begin{equation}
    P_{\rm odd} = \frac{1-\Tr\left[\rho'^2\right]}{2}.
\end{equation}
(See Appendix~\ref{sec:swap_test} for the detailed derivation.) Meanwhile, Eq.~\eqref{eq:twirled_state} shows that the purity of $\rho'$ is $\Tr[\rho'^2]=1-2\epsilon + O(\epsilon^2)$, leading to $P_{\rm odd}=\epsilon+O(\epsilon^2)$. Therefore, applying the Bell measurement scheme on the twirled state $\rho'$ enables sampling from ${\rm Bern}(\epsilon + O(\epsilon^2))$ and thus achieving $O(1/\epsilon)$ sample complexity. More precisely, applying Chernoff bound to the resulting binomial distribution yields:

\begin{theorem} 
    For a $n$-qubit state $\ket{\psi} \in \mathcal{H}$, denote $\mathcal{H}=\mathcal{H}_0\oplus \dots \oplus \mathcal{H}_k$, where $\mathcal{H}_j$ are irreducible representations of $G_\psi$ and $\mathcal{H}_0={\rm span}\{\ket{\psi}\}$. If $\mathcal{H}_{0} \not\simeq \mathcal{H}_j$ for all $j=1,\dots,k$, then for $0<r<1$ and for any state $\rho$ with infidelity $\epsilon$, Bell measurement scheme outputs an estimator $\hat{\epsilon}$ satisfying $|\epsilon-\hat{\epsilon}|\le (r+O(\epsilon))\epsilon$ with probability at least $1-\delta$, using $N=O(\log(1/\delta)\cdot1/r^2\epsilon)$ copies of $\rho$.
    \label{thm:suff_two-copy}
\end{theorem}

We remark that Theorem~\ref{thm:suff_two-copy} applies to a wide range of magic states, including both single-qubit magic states: $\ket{T}$ and $\ket{H}=\sqrt{1/2}(\ket{0}+e^{i\pi/4}\ket{1})$, and multi-qubit magic states:
\begin{align}
    \ket{\rm CZ}&=\frac{1}{\sqrt{3}}(\ket{00}+\ket{01}+\ket{10}),\\
    \ket{\rm CCZ}&=\frac{1}{2}(\ket{00+}+\ket{01+}+\ket{10+}+\ket{11-}),
\end{align}
where we denote $\ket{\pm}=(\ket{0}\pm\ket{1})/\sqrt{2}$ (see Appendix~\ref{sec:twirling}).

While Bell measurements on two copies of $\rho$ provide an efficient benchmarking protocol, preparing two high-fidelity magic states simultaneously might be challenging, especially when the magic state distillation has a low success probability and there is no reliable way to store the distilled states. In such cases, single-copy benchmarking schemes are desirable, and we find that leveraging multi-qubit magic states enables efficient single-copy benchmarking protocol. The key idea is that certain multi-qubit magic states possess orthogonal stabilizer states whose overlaps with the twirled state $\rho'$ can specify the infidelity. To be specific, as seen above, for $i=0,1,\dots, k$, $\rho_{ii} = \lambda_i\Pi_i$ while $\lambda_0 = 1-\epsilon$. Since $\rho'$ has a unit trace, we have
\begin{equation}
    \epsilon=\sum_{i=1}^k \lambda_i \cdot \dim\mathcal{H}_i.
\end{equation}
Therefore, estimating $\lambda_1,\dots,\lambda_k$ enables estimation of $\epsilon$. For $j=1,\dots, k$, if each $\mathcal{H}_j$ contains a stabilizer state $\ket{s_j}$, then the overlap $\langle s_j|\rho'|s_j\rangle=\lambda_j$ allows efficient estimation of $\lambda_j$, enabling an efficient single-copy benchmarking protocol.

\begin{theorem} 
    For a $n$-qubit state $\ket{\psi} \in \mathcal{H}$, with $n$ being constant larger than 1, denote $\mathcal{H}=\mathcal{H}_0\oplus \mathcal{H}_1\oplus \dots \oplus \mathcal{H}_k$, where $\mathcal{H}_j$ are irreducible representations of $G_\psi$ and $\mathcal{H}_0={\rm span}\{\ket{\psi}\}$. Suppose there exist stabilizer states $\ket{s_1}\in\mathcal{H}_1,\dots,\ket{s_k}\in\mathcal{H}_k$. Then, there is a single-copy scheme that, for $0<r<1$ and any state $\rho$ with infidelity $\epsilon$, outputs an estimator $\hat{\epsilon}$ satisfying $|\epsilon-\hat{\epsilon}|\le r\epsilon$ with probability at least $1-\delta$, using $N=O(\log(1/\delta)\cdot1/r^2\epsilon)$ copies of $\rho$.
    \label{thm:suff_single-copy}
\end{theorem}

\begin{figure*}[t]
    \centering
    \includegraphics[width=\linewidth]{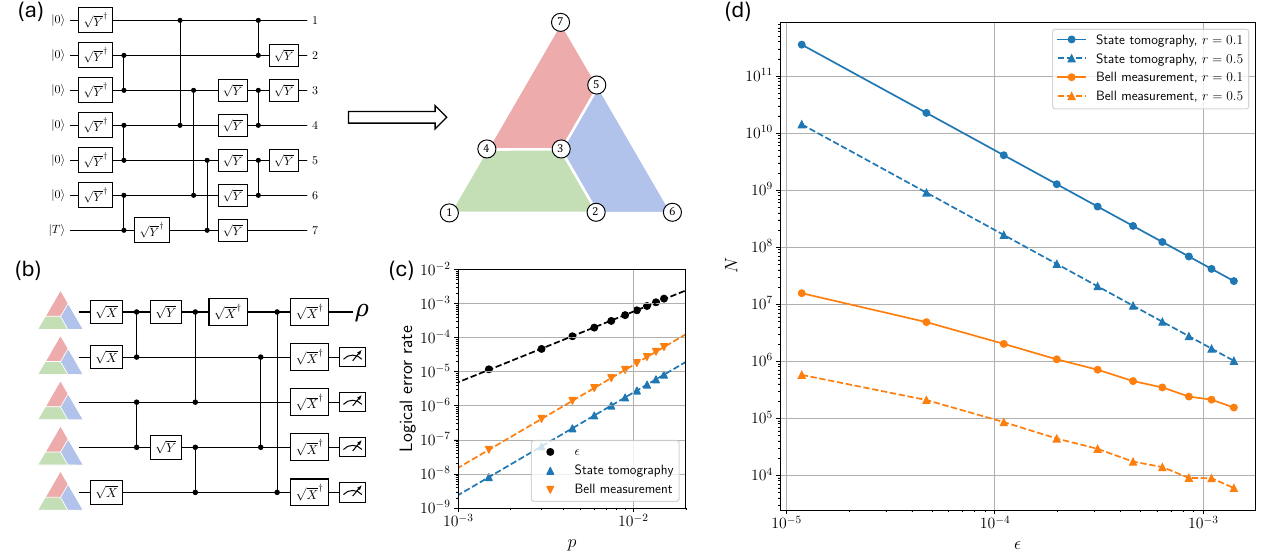}
    \caption{\label{fig:fig2}(a) Circuit for encoding a physical $\ket{T}$ state into the [[7, 1, 3]] color code, generating a noisy logical magic state. (b) Circuit for performing 5-to-1 magic-state distillation (MSD) at the logical level, resulting in a higher-fidelity logical state $\rho$. Both circuits in (a) and (b) follow the methodology of Ref.~\cite{rodriguezExperimentalDemonstrationLogical2024}. (c) Infidelity $\epsilon$ of the state $\rho$ and the logical error rates from benchmarking circuits. The dashed lines (black, blue, and orange) indicate power-law fits, given by $y=8.11x^{2.07}$, $y=2.52x^{3.01}$, and $y=16.8x^{3.01}$, respectively. (d) Sampling overhead $N$ for the standard state tomography and the Bell measurement scheme with multiplicative accuracies $r = 0.1$ and $r = 0.5$, with $68\%$ confidence.}
\end{figure*}

Important applications of Theorem~\ref{thm:suff_single-copy} are benchmarking $\ket{\rm CZ}$ and $\ket{\rm CCZ}$ states. First, $\ket{\rm CZ}$ state decomposes the Hilbert space as $\mathcal{H}=\mathcal{H}_0 \oplus \mathcal{H}_1 \oplus \mathcal{H}_2$ with the projectors $\Pi_1 = \ketbra{11}{11}$ and $\Pi_2 = I-\ketbra{\rm CZ}{\rm CZ}-\ketbra{11}{11}$. Then, the corresponding stabilizer states $\ket{11} \in \mathcal{H}_1$ and $(\ket{01} - \ket{10})/\sqrt{2} \in \mathcal{H}_2$ allow efficient benchmarking. Similarly, $\ket{\rm CCZ}$ state decomposes the Hilbert space as $\mathcal{H}=\mathcal{H}_0 \oplus \mathcal{H}_1$ with $\Pi_1 = I-\ketbra{\rm CCZ}{\rm CCZ}$. The stabilizer state $\ket{00+} \in \mathcal{H}_1$ enables an efficient single-copy fidelity estimation (see Appendix~\ref{sec:twirling} for a detailed analysis). For an arbitrary multi-qubit magic state, we provide an explicit algorithm to check whether the condition of Theorem~\ref{thm:suff_single-copy} is satisfied (see Appendix~\ref{sec:twirling}).

The benchmarking schemes of Theorem~\ref{thm:suff_two-copy} and Theorem~\ref{thm:suff_single-copy} achieve the sample complexity $O(1/r^2\epsilon)$ with a fixed success probability $1-\delta$. We remark that this scaling is optimal in the sense that even when allowing arbitrary joint measurements on multiple copies of $\rho$, any benchmarking scheme must consume $\Omega(1/r^2\epsilon)$ samples:

\begin{theorem}
    Let $\ket{\psi}$ be an $n$-qubit state. Suppose there exists a benchmarking scheme that allows arbitrary joint measurements on multiple copies of $\rho$ and, for any noisy state $\rho$ with infidelity $\epsilon$, outputs an estimator $\hat{\epsilon}$ satisfying $|\hat{\epsilon} - \epsilon| \leq r\epsilon$ with probability at least $2/3$, while consuming $N$ copies of $\rho$ in total. Then, $N = \Omega(1/r^2\epsilon)$.
    \label{thm:optimality}
\end{theorem}

\section{Numerical results for experimental realization}

Finally, we present numerical results demonstrating the effectiveness of our protocols in realistic experimental settings. Our simulations closely follow the recent experimental demonstration of a 5-to-1 logical $\ket{T}$ magic-state distillation (MSD)~\cite{rodriguezExperimentalDemonstrationLogical2024}.

Specifically, we consider a scenario where we generate a high-fidelity logical $\ket{T}$ state via MSD, encoded in the [[7, 1, 3]] color code, which supports transversal Clifford gates and Pauli measurements. Initially, we encode a physical $\ket{T}$ state into the [[7, 1, 3]] color code using the circuit shown in Fig.~\ref{fig:fig2}(a), creating a noisy logical magic state. Subsequently, we apply a 5-to-1 MSD protocol~\cite{bravyiUniversalQuantumComputation2005} at the logical level, as illustrated in Fig.~\ref{fig:fig2}(b), yielding a higher-fidelity logical magic state $\rho$.

To benchmark the resulting logical state $\rho$, we first apply the twirling operation at the logical level by transversal Clifford gates. Then, we perform two types of benchmarking schemes: standard state tomography and our proposed Bell measurement scheme, both implemented at the logical level. Finally, we perform error detection by analyzing syndromes derived from physical measurement outcomes and post-select syndrome-free runs, thereby effectively projecting the distilled state $\rho$ into the logical space. See Appendix~\ref{sec:numerical_appen} for details regarding the numerical simulation.

Reflecting typical conditions in state-of-the-art devices~\cite{BluvsteinLogical2024, rodriguezExperimentalDemonstrationLogical2024}, each operation at the physical level is subject to noise parameterized by $p$: single-qubit gates are followed by single-qubit depolarizing channels with average fidelity $1 - p/5$; two-qubit gates are followed by two-qubit depolarizing channels with fidelity $1 - p$; state preparations are followed by single-qubit depolarizing noise with average fidelity $1 - p/2$; and measurement outcomes are flipped with probability $p/2$.

We confirmed that the majority of logical errors originate from the non-fault-tolerant encoding circuit [Fig.~\ref{fig:fig2}(a)], while the 5-to-1 MSD, which is fault-tolerant, suppresses the infidelity $\epsilon$ of $\rho$ up to the second order of $p$ [Fig.~\ref{fig:fig2}(c)]. Meanwhile, both benchmarking circuits of state tomography and Bell measurement utilize only transversal Clifford gates, ensuring fault-tolerance. Since the [[7, 1, 3]] color code can detect up to two errors, logical errors from MSD and benchmarking circuits are suppressed to the third order of $p$, rendering them negligible compared to $\epsilon$ [Fig.~\ref{fig:fig2}(c)].

Numerical results presented in Fig.~\ref{fig:fig2}(d) show that our Bell measurement scheme maintains its scaling advantage even under realistic imperfections, significantly reducing the required number of samples compared to standard state tomography. Specifically, at noise levels around $p = 0.01$, which are already achievable by modern experiments~\cite{rodriguezExperimentalDemonstrationLogical2024, BluvsteinLogical2024, Kim2023EvidenceUtility, mckay2023benchmarkingquantumprocessorperformance, GoogleWillow2025, decross2024computationalpowerrandomquantum, AkhtarIonQ2023}, our method reduces the sampling overhead by over two orders of magnitude. Additional numerical results for the single-copy multi-qubit scheme, presented in Appendix~\ref{sec:numerical_appen_additional}, lead to the same conclusion. Together, these results strongly support the practical feasibility of our protocol for near-term experimental implementation.

\begin{acknowledgments}
We thank Dolev Bluvstein, Mikhail Lukin, Qian Xu, and Hengyun Zhou for helpful discussions. S.L. and S.C. are grateful to Madhur Tulsiani for the course ``Information and Coding Theory" (TTIC 31200) at the University of Chicago, which inspired rigorous proofs presented in this work. S.L., M.Y., S.C., and L.J. acknowledge support from the ARO(W911NF-23-1-0077), ARO MURI (W911NF-21-1-0325), AFOSR MURI (FA9550-21-1-0209, FA9550-23-1-0338), DARPA (HR0011-24-9-0359, HR0011-24-9-0361), NSF (ERC-1941583, OMA-2137642, OSI-2326767, CCF-2312755, OSI-2426975), and the Packard Foundation (2020-71479). K.T. is supported by the Program for Leading Graduate Schools (MERIT-WINGS), JST BOOST Grant Number JPMJBS2418, JST ASPIRE Grant Number JPMJAP2316, and JST ERATO Grant Number JPMJER2302. S.L. is partially supported by the Kwanjeong Educational Foundation. We are also grateful for the support of the University of Chicago’s Research Computing Center for assistance with the calculations carried out in this work.
\end{acknowledgments}

\bibliography{references}


\onecolumngrid
\appendix

\section{\label{sec:reprethry}Preliminaries on representation theory}
In this section, we briefly review key results from representation theory utilized in this work, closely following Chapter 1 of Ref.~\cite{fultonRepresentationTheory2004}, but restricting to representations of finite subgroups of the unitary group. For a comprehensive explanation, see Ref.~\cite{fultonRepresentationTheory2004}.

Let $\mathcal{H}$ be a finite-dimensional vector space (e.g., $n$-qubit Hilbert space) and $\mathcal{U}$ be the corresponding unitary group. For a finite subgroup $G\le \mathcal{U}$, we first define an invariant subspace:

\begin{definition}
    A subspace $\mathcal{H}' \le \mathcal{H}$ is \emph{invariant under $G$} if it is invariant under the action of any $U \in G$, i.e.,
    \begin{equation}
        U\ket{\psi} \in \mathcal{H}',
    \end{equation}
    for all $\ket{\psi} \in \mathcal{H}'$.
\end{definition}

Here, note that the trivial subspace $\{0\}$ is also invariant under $G$. Next, we introduce $G$-linear maps and equivalence between invariant subspaces:
\begin{definition}
    Given two invariant subspaces $\mathcal{H}_1, \mathcal{H}_2 \le \mathcal{H}$ under $G$, a linear map $\varphi:\mathcal{H}_1 \rightarrow \mathcal{H}_2$ is a \emph{$G$-linear map} if
    \begin{equation}
        \varphi \circ U = U\circ\varphi
    \end{equation}
    for all $U \in G$. Here, $U$ on the left-hand side acts on $\mathcal{H}_1$ and $U$ on the right-hand side acts on $\mathcal{H}_2$. Equivalently, the following diagram commutes for all $U\in G$:
    \begin{equation}
        \begin{tikzcd}[row sep=large, column sep=large]
        \centering
        \mathcal{H}_1 \arrow[r, "\varphi"] \arrow[d, "U"] & \mathcal{H}_2 \arrow[d, "U"] \\
        \mathcal{H}_1 \arrow[r, "\varphi"] & \mathcal{H}_2
        \end{tikzcd}
    \end{equation}
    Furthermore, if there exists a $G$-linear map $\varphi:\mathcal{H}_1 \rightarrow \mathcal{H}_2$ that is one-to-one and onto, we say $\mathcal{H}_1$ and $\mathcal{H}_2$ are \emph{equivalent}, denoted by $\mathcal{H}_1\simeq\mathcal{H}_2$.
\end{definition}

With a $G$-linear map $\varphi:\mathcal{H}_1\rightarrow\mathcal{H}_2$, we have
\begin{align}
    \varphi(\ket{\psi}) = 0
    &\Leftrightarrow (U\circ\varphi)(\ket{\psi}) = 0\\
    &\Leftrightarrow \varphi(U\ket{\psi}) = 0,
\end{align}
for all $\ket{\psi} \in \mathcal{H}_1$ and $U\in G$. Therefore, the kernel of a $G$-linear map is invariant under $G$. Similarly, the image of a $G$-linear map is also invariant under $G$. With these relations, we can decompose $\mathcal{H}$ with a direct sum of invariant subspaces:

\begin{proposition}[Adapted from Prop.~1.5 in \cite{fultonRepresentationTheory2004}]
    If $\mathcal{H}' \le \mathcal{H}$ is an invariant subspace under $G$, there exists a complementary invariant subspace $\mathcal{H}'^\perp$ such that $\mathcal{H}=\mathcal{H}'\oplus\mathcal{H}'^\perp$.
    \label{prop:decomposition}
\end{proposition}

\begin{proof}
    Let $\varphi_0:\mathcal{H} \rightarrow \mathcal{H'}$ be a projection onto $\mathcal{H'}$ given by the decomposition $\mathcal{H}=\mathcal{H}'\oplus\mathcal{H}''$. Define another map $\varphi:\mathcal{H} \rightarrow \mathcal{H'}$,
    \begin{equation}
        \varphi=\sum_{U\in G}U\circ\varphi_0\circ U^\dagger.
    \end{equation}
    Since $\varphi \circ V = \sum_{U\in G}U\circ\varphi_0\circ (U^\dagger V) = \sum_{U\in G}(VU)\circ\varphi_0\circ U^\dagger=V\circ\varphi$ for all $V \in G$, $\varphi$ is a $G$-linear map. Therefore, $\ker\varphi$ is a subspace of $\mathcal{H}$ invariant under $G$ and complementary to $\mathcal{H'}$.
\end{proof}

An irreducible representation is then defined as follows:

\begin{definition}
    A subspace $\mathcal{H}' \le \mathcal{H}$ is an \emph{irreducible representation} of $G$ if:
    \begin{enumerate}[label=(\roman*)]
        \item $\mathcal{H}'$ is invariant under $G$.
        \item $\mathcal{H}'$ has no non-trivial proper invariant subspaces under $G$.
    \end{enumerate}
\end{definition}

In other words, the action of $G$ on an irreducible representation cannot be decomposed further. Proposition~\ref{prop:decomposition} immediately implies:

\begin{corollary}
    $\mathcal{H}$ can be decomposed as a direct sum of irreducible representations.
\end{corollary}

Finally, we present Schur’s lemma:

\begin{theorem}[Schur's lemma]
    If $\mathcal{H}_1, \mathcal{H}_2 \le \mathcal{H}$ are irreducible representations of $G$ and $\varphi:\mathcal{H}_1 \rightarrow \mathcal{H}_2$ is a $G$-linear map, then
    \begin{enumerate}[label=(\roman*)]
        \item Either $\varphi$ is an isomorphism, or $\varphi=0$.
        \item If $\mathcal{H}_1 = \mathcal{H}_2$, then $\varphi$ is proportional to the identity, i.e., $\varphi= \lambda\cdot I$ for some $\lambda \in \mathbb{C}$.
    \end{enumerate}
\end{theorem}
\begin{proof}
    (i) Since $\ker\varphi \le \mathcal{H}_1$ and ${\rm im}~\varphi\le \mathcal{H}_2$, the irreducibilities of $\mathcal{H}_1$ and $\mathcal{H}_2$ implies either (1) $\ker\varphi = \{0\}$ and ${\rm im}~\varphi = \mathcal{H}_2$ (i.e., $\varphi$ is an isomorphism), or (2) $\ker\varphi = \mathcal{H}_1$ and ${\rm im}~\varphi = \{0\}$ (i.e., $\varphi=0$).

    (ii) Let $\lambda \in \mathbb{C}$ be an eigenvalue of $\varphi$, so that $\ker(\varphi - \lambda \cdot I) \ne \{0\}$. As $\varphi - \lambda \cdot I$ is also a $G$-linear map, $\varphi - \lambda \cdot I=0$ by (i), and it leads to $\varphi=\lambda \cdot I$.
\end{proof}

\section{\label{sec:swap_test}Bell measurement and SWAP test}

\begin{figure}
    \includegraphics[width=0.3\linewidth]{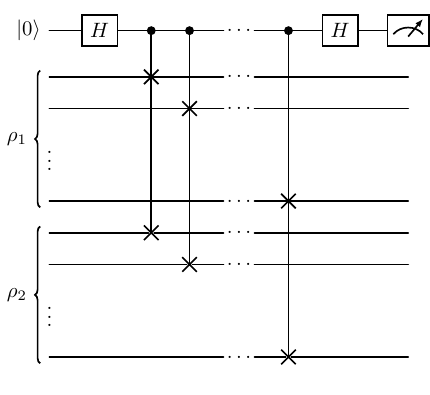}
    \caption{The SWAP test circuit for measuring the Hilbert-Schmidt inner product $\Tr[\rho_1 \rho_2]$.}
    \label{fig:SWAP_test_circuit}
\end{figure}

The SWAP test, originally proposed for determining the overlap between two quantum states~\cite{buhrmanQuantumFingerprinting2001, FilipOverlap}, can be utilized to measure the purity of a state, or more generally, the Hilbert-Schmidt inner product between two quantum states. In this section, we explain how the Bell measurement circuit effectively implements the SWAP test, adapted from Ref.~\cite{garcia-escartinSwapTestHongOuMandel2013}.

Let $\rho$ be an $n$-qubit state. The purity can be measured by preparing two copies of $\rho$ and performing the SWAP test circuit.
More generally, the SWAP test can be used to measure the Hilbert-Schmidt inner product between two quantum states $\rho_1$ and $\rho_2$, i.e., $\Tr[\rho_1 \rho_2]$. (Note that the purity is a special case with $\rho_1 = \rho_2 = \rho$.)
This circuit consists of two input states $\rho_1$ and $\rho_2$ with one ancilla qubit initialized in $\ket{0}$. We first apply a Hadamard gate on the ancilla qubit and apply a controlled SWAP gate---a gate applies SWAP gate on the two copies of $\rho$ if the ancilla is in $\ket{1}$ state, and does nothing otherwise. This controlled SWAP gate can be realized by applying controlled SWAP gates on all pairs of qubits of $\rho_1$ and $\rho_2$, as seen in Fig.~\ref{fig:SWAP_test_circuit}. After that, we apply another Hadamard gate to the ancilla and finally measure it.
Then, the probability of measuring $1$ is given by
\begin{equation}
\begin{aligned}
    P(1) &= \Tr\left[(\bra{-}\otimes I \otimes I){\rm CSWAP}\left(\ketbra{+}{+}\otimes \rho_1 \otimes \rho_2\right){\rm CSWAP}^\dagger (\ket{-}\otimes I \otimes I)\right]\\
    &= \frac{1}{2}\Tr\big[(\bra{-}\otimes I \otimes I)
    \big(\ketbra{0}{0}\otimes \rho_1 \otimes \rho_2 + \ketbra{0}{1}\otimes (\rho_1 \otimes \rho_2 {\rm SWAP}) \\ &\qquad\qquad\qquad\qquad\qquad+ \ketbra{1}{0}\otimes ( {\rm SWAP} \rho_1 \otimes \rho_2) + \ketbra{1}{1}\otimes \rho_2 \otimes \rho_1 \big)
    (\ket{-}\otimes I \otimes I)\big]\\
    &=\frac{1}{4}\left(\Tr[\rho_1 \otimes \rho_2] - \Tr[{\rm SWAP}\rho_1\otimes\rho_2] - \Tr[\rho_1\otimes\rho_2{\rm SWAP}] + \Tr[\rho_2 \otimes \rho_1]\right)\\
    &=\frac{1-\Tr[\rho_1 \rho_2]}{2},
\end{aligned}
\end{equation}
where we denote ${\rm CSWAP}$ as the controlled-SWAP gate. Therefore, repeating this circuit enables the estimation of the Hilbert-Schmidt inner product $\Tr[\rho_1 \rho_2]$.

\begin{figure}
    \centering
    \includegraphics[width=\linewidth]{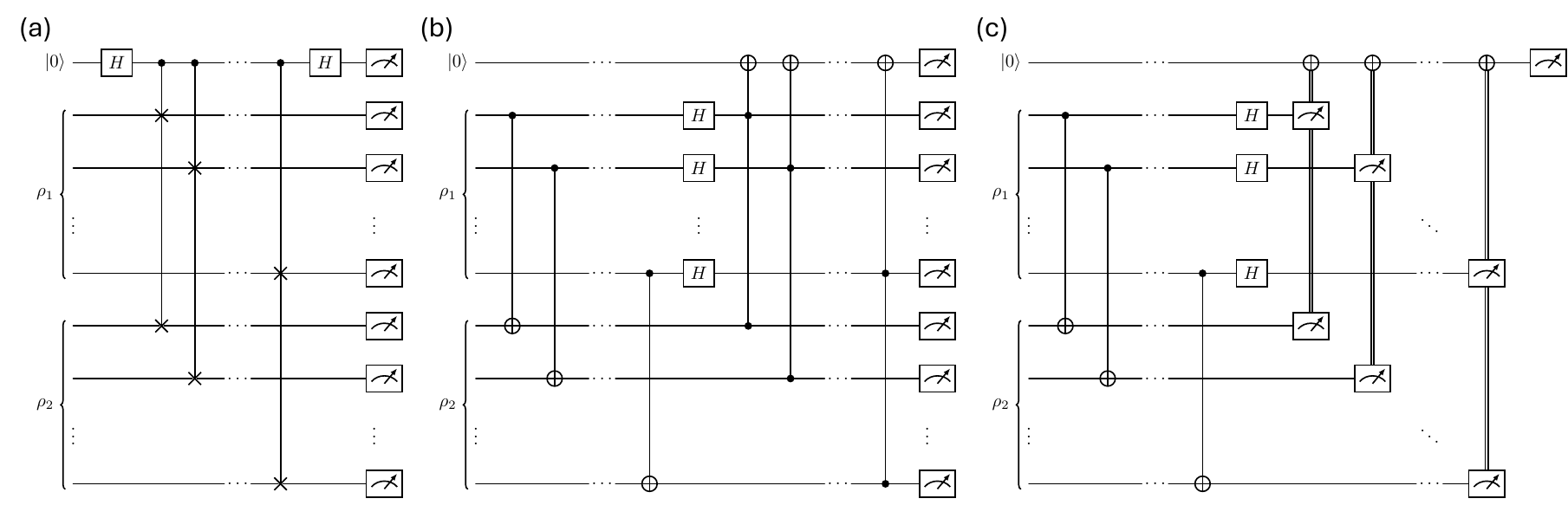}
    \caption{Equivalence of the SWAP test and the Bell measurement circuits. (a) The original SWAP test circuit with additional measurements at the end. (b) The circuit equivalent to the SWAP test circuit. (c) The Bell measurement circuit followed by feedforward channels.}
    \label{fig:equiv_swap_bell}
\end{figure}

Now, we explain how the Bell measurement circuit realizes the SWAP test thereby enabling the purity estimation. First, by the principle of implicit measurement (see, e.g., Ref.~\cite{Nielsen_Chuang_2010}), the SWAP test circuit is unaffected by measuring the remaining $2n$ qubits in the end [Fig.~\ref{fig:equiv_swap_bell}(a)]. Ref.~\cite{garcia-escartinSwapTestHongOuMandel2013} showed that this SWAP test circuit is equivalent to the following circuit described in Fig.~\ref{fig:equiv_swap_bell}(b), by decomposing SWAP gates with multiple rounds of CNOT gates. By the principle of deferred measurement (see, e.g., Ref.~\cite{Nielsen_Chuang_2010}), the circuit described in Fig.~\ref{fig:equiv_swap_bell}(c) which replaces the controlled-controlled-not gates with measurement and feedforward channels gives the same measurement statistics. Here, each measurement and feedforward channel activates $X$ gate on the ancilla when the two corresponding measurement outcomes are 1.

Note that the circuit in Fig.~\ref{fig:equiv_swap_bell}(c) is the Bell measurement circuit followed by the post-processing. The measurement outcome of the ancilla is 1 if and only if an odd number of $X$ gates are activated. Denoting the measurement outcome of $2n$ qubits corresponding to the two copies of $\rho$ as $x\in\{0,1\}^{2n}$, this condition is identical that $\sum_{i=1}^n x_i \cdot x_{i+n}$ is an odd number. Since the probability of measuring $1$ for ancilla is $\frac{1-\Tr[\rho_1 \rho_2]}{2}$, we have
\begin{equation}
    P_{\rm odd} = \frac{1-\Tr\left[\rho_1 \rho_2\right]}{2},
\end{equation}
where $P_{\rm odd}$ denotes the probability of $\sum_{i=1}^n x_i \cdot x_{i+n}$ being an odd number. Therefore, the Bell measurement circuit effectively realizes the SWAP test and estimates the Hilbert-Schmidt inner product $\Tr[\rho_1 \rho_2]$.

Here, we make a few remarks on the Bell measurement scheme. First, the Bell measurement scheme works even the two copies of the input magic states $\rho_1$ and $\rho_2$ are different. Let $\rho_1$ and $\rho_2$ be two $n$-qubit states with infidelities $\epsilon_1$ and $\epsilon_2$ to a target magic state $\ket{\psi}$, respectively. Then, the twirling operation converts them into $\rho_1'$ and $\rho_2'$, where we can write
\begin{align}
    \rho_1' &= (1-\epsilon_1)\ketbra{\psi}{\psi} + \epsilon_1 \sigma_1,\\
    \rho_2' &= (1-\epsilon_2)\ketbra{\psi}{\psi} + \epsilon_2 \sigma_2,
\end{align}
where $\sigma_1$ and $\sigma_2$ are some mixed states that have no overlap with $\ket{\psi}$. Then the output of the Bell measurement scheme is given by
\begin{equation}
    P_{\rm odd} = \frac{1-\Tr[\rho_1' \rho_2']}{2} = \frac{\epsilon_1 + \epsilon_2 - \epsilon_1 \epsilon_2 (1 + \Tr[\sigma_1 \sigma_2])}{2}.
\end{equation}
Therefore, if $\epsilon_1$ and $\epsilon_2$ are sufficiently small, $P_{\rm odd}$ approximates $(\epsilon_1 + \epsilon_2)/2$. Consequently, the Bell measurement scheme efficiently estimates the average infidelity $(\epsilon_1 + \epsilon_2)/2$.

Second, we also remark that only one of the two copies needs to be twirled. To see this, assume that only $\rho_1$ is twirled into $\rho_1'$, while $\rho_2$ is not twirled. Specifically, let $G_{\psi}$ be the twirling group associated with $\ket{\psi}$. Then, the Hilbert-Schmidt inner product between $\rho_1'$ and $\rho_2$ is given by
\begin{align}
    \Tr[\rho_1' \rho_2] &= \Tr\left[\left(\frac{1}{|G_{\psi}|}\sum_{U\in G_{\psi}} U\rho_1 U^\dagger\right)\rho_2\right]\\
    &= \Tr\left[\left(\frac{1}{|G_{\psi}|^2}\sum_{V\in G_{\psi}}\sum_{U\in G_{\psi}} V^\dagger U\rho_1 U^\dagger V\right)\rho_2\right]\\
    &= \frac{1}{|G_{\psi}|^2}\sum_{U, V\in G_{\psi}}\Tr[U\rho_1 U^\dagger V\rho_2 V^\dagger]\\
    &= \Tr[\rho_1' \rho_2'],
\end{align}
where $\rho_2' = \frac{1}{|G_{\psi}|}\sum_{V\in G_{\psi}} V\rho_2 V^\dagger$ is the twirled state of $\rho_2$. In other words, twirling only one copy has the same effect as twirling both copies. Therefore, the Bell measurement scheme only requires one of the input states to be twirled.

\section{\label{sec:proofs}Proofs of the main theorems}

In this section, we provide the deferred proofs of the main theorems. To this end, we first provide the Chernoff bounds for binomial distributions:
\begin{lemma}[Chernoff bounds for binomial distribution, adapted from Ref.~\cite{mitzenmacherProbabilityComputingRandomized2005}]
    Let $X_1,\dots,X_N \sim {\rm Bern}(q)$ be independent and identically distributed random variables for $q>0$. Denoting $\hat{q}=\frac{1}{N}\sum_{i=1}^NX_i$, the following inequalities hold:
    \begin{enumerate}[label=(\roman*)]
        \item for $\delta > 0$,
        \begin{equation}
            \mathbb{P}\left[\hat{q}-q\ge \delta q\right] <  \left(\frac{e^\delta}{(1+\delta)^{(1+\delta)}}\right)^{Nq},
        \end{equation}
        \item for $0<\delta<1$,
        \begin{equation}
            \mathbb{P}\left[\hat{q}-q\ge \delta q\right] < e^{-Nq\delta^2/3},
        \end{equation}
        \item for $\delta >0$,
        \begin{equation}
            \mathbb{P}\left[\hat{q}-q\le -\delta q\right] < e^{-Nq\delta^2/2}.
        \end{equation}
    \end{enumerate}
    \label{lem:chernoff}
\end{lemma}
Note that (i) and (iii) imply $\mathbb{P}\left[|\hat{q}-q|\ge \delta q\right] < 2e^{-Np\delta^2/3}$ for $0<\delta<1$. In addition, we can also show that
\begin{equation}
    \frac{e^\delta}{(1+\delta)^{(1+\delta)}} \le e^{-\frac{\delta^2}{2+\delta}},
    \label{eq:chernoff_application}
\end{equation}
for $\delta>1$. To this end, note that this relation is equivalent to
\begin{equation}
    f(\delta) = \delta - (1+\delta)\ln(1+\delta) + \frac{\delta^2}{2+\delta} \le 0.
\end{equation}
This is shown by $f'(\delta)=-\ln(1+\delta)+\frac{2\delta}{2+\delta} \le 0$ and $f(0) = 0$. Combining the first bound of Lemma~\ref{lem:chernoff} and Eq.~\ref{eq:chernoff_application}, we have the following:

\begin{lemma}
    Let $X_1,\dots,X_N \sim {\rm Bern}(q)$ be independent and identically distributed random variables for $q>0$. Denoting $\hat{q}=\frac{1}{N}\sum_{i=1}^NX_i$, we have
    \begin{equation}
        \mathbb{P}\left[\hat{q}-q\ge \delta q\right] <  e^{-\left(\frac{\delta^2}{2+\delta}\right)Nq}.
    \end{equation}
    \label{lem:chernoff_2}
\end{lemma}

With these concentration bounds of Lemmas~\ref{lem:chernoff} and \ref{lem:chernoff_2}, the rest of this section proves Theorems~\ref{thm:no-go}, \ref{thm:suff_two-copy}, \ref{thm:suff_single-copy}, and \ref{thm:optimality}.

\begingroup
  \renewcommand{\thetheorem}{\ref{thm:no-go}}
  \begin{theorem}[restated]
    \label{thm:no-go-restated}
    Let $\ket{\psi}$ be an $n$-qubit state having no orthogonal stabilizer state. Consider any single-copy benchmarking protocol producing an estimator $\hat{\epsilon}$ of the infidelity $\epsilon = 1 - \langle\psi|\rho|\psi\rangle$, satisfying the multiplicative error $|\hat{\epsilon}-\epsilon| \leq r\epsilon$ with probability at least $2/3$, for an arbitrary input state $\rho$. Such a protocol necessarily requires $N = \Omega\left(1/r^2\epsilon^2\right)$ copies of $\rho$.
  \end{theorem}
\endgroup

\begin{proof}
    Let $\ket{\psi}$ be an $n$-qubit state and $\hat{\epsilon} = \hat{\epsilon}(X_1,\dots,X_N)$ be the estimator of the input state $\rho$'s infidelity with $N$ samples. Here, $X_j \in \{0,1\}^n$ is the measurement outcome from $j$-th copy. Let us further assume that given the input state's infidelity $\epsilon$ and the multiplicative precision $r$, $\hat{\epsilon}$ guarantees to be $|\hat{\epsilon} - \epsilon|\le r\epsilon$ with probability at least $2/3$, if $N \ge N_{\epsilon,r}$. With this setup, the goal of the proof is to show that $N_{\epsilon,r}=\Omega(1/r^2\epsilon^2)$.
    
    We first show that using this estimator $\hat{\epsilon}$, we can distinguish the following states:
    \begin{align}
        \rho_{0} &= (1-\epsilon)\ketbra{\psi}{\psi} + \frac{\epsilon}{2^n -1}(I-\ketbra{\psi}{\psi}), \label{eq:rho_1}\\
        \rho_{1} &= (1-(1-2r)\epsilon)\ketbra{\psi}{\psi} + \frac{(1-2r)\epsilon}{2^n -1}(I-\ketbra{\psi}{\psi}),
        \label{eq:rho_2}
    \end{align}
    with probability at least 2/3. Here, we take $r<1/2$. Specifically, note that the infidelities of $\rho_0$ and $\rho_1$ are $\epsilon$ and $(1-2r)\epsilon$, respectively. We choose $N =\max\{N_{\epsilon,r},N_{(1-2r)\epsilon, r}\}$, and let $P$ and $Q$ be the distributions of the measurement outcomes $(X_1,\dots,X_N)$ for $\rho_0$ and $\rho_1$, respectively. Then, by the assumption on $\hat{\epsilon}$, we have
    \begin{align}
        \underset{(X_1,\dots,X_N)\sim P}{\mathbb{P}}\left[\hat{\epsilon}(X_1,\dots,X_N) \le (1+r)(1-2r)\epsilon\right] \ge 2/3,\\
        \underset{(X_1,\dots,X_N)\sim Q}{\mathbb{P}}\left[\hat{\epsilon}(X_1,\dots,X_N) \ge (1-r)\epsilon\right] \ge 2/3.\\
    \end{align}
    Since $(1+r)(1-2r)\epsilon \le (1-r)\epsilon$ the function $T:\{0,1\}^{n} \mapsto \{0,1\}$,
    \begin{equation}
        T(X_1,\dots,X_N) =
        \begin{cases}
            0, \quad&\text{if }\hat{\epsilon}(X_1,\dots,X_N) \ge (1-r)\epsilon,\\
            1, \quad&\text{if }\hat{\epsilon}(X_1,\dots,X_N) < (1-r)\epsilon,
        \end{cases}
    \end{equation}
    distinguishes $\rho_0$ and $\rho_1$ with probability at least $2/3$.
    
    Now, we show that for such $T$ to exist, $N=\Omega(1/r^2\epsilon^2)$. To this end, note that the total variation distance $d_{\rm TV}(R_1,R_2)$, for any probability distributions $R_1$ and $R_2$ over the sample space $\Omega$, should satisfy $d_{\rm TV}(R_1,R_2)=\sup_{A\subset \Omega}\left\lvert\underset{x \sim R_1}{\mathbb{P}}[x\in A]-\underset{x \sim R_2}{\mathbb{P}}[x\in A]\right\rvert$  Therefore,
    \begin{align}
        d_{\rm TV}(P,Q) &\ge \underset{(X_1,\dots,X_N) \sim P}{\mathbb{P}}[T(X_1,\dots,X_N)=0]-\underset{(X_1,\dots,X_N) \sim Q}{\mathbb{P}}[T(X_1,\dots,X_N)=0]\\
        &\ge \frac{1}{3}.
        \label{eq:tvd}
    \end{align}
    By Pinsker's inequality, we have
    \begin{equation}
        D_{\rm KL}(Q\|P) \ge \frac{2}{\ln 2}\left(d_{\rm TV}(P,Q)\right)^2,
    \end{equation}
    where $D_{\rm KL}(R_1\| R_2)=\sum_{x \in \Omega}R_1(x)\log\left(\frac{R_1(x)}{R_2(x)}\right)$ denotes the Kullback-Leibler divergence (KL divergence). Therefore, we have
    \begin{equation}
        D_{\rm KL}(Q\|P) \ge \frac{2}{9\ln 2}.
        \label{eq:pinsker}
    \end{equation}
    
    Meanwhile, the Clifford gate $U_i$ for the $i$-th measurement can be chosen based on the previous measurement outcomes $x_1,\dots,x_{i-1}$, so the distribution of $X_i$ depends on those previous outcomes. Therefore, the probabilities of measuring $X_i = x_i$ are
    \begin{align}
        P(X_i|X_1=x_1,\dots X_{i-1}=x_{i-1}),\\
        Q(X_i|X_1=x_1,\dots X_{i-1}=x_{i-1}).
    \end{align}
    Let us denote these distributions as $P_i(x_i|x_{<i})$ and $Q_i(x_i|x_{<i})$, respectively. Further, let $R(x_i)$ be the output distribution of the ideal magic state, i.e.,
    \begin{equation}
        R_i(x_i|x_{<i}) = \left\vert\langle x|U_i|\psi \rangle\right\vert^2.
    \end{equation}
    With these notions and Eqs.~\eqref{eq:rho_1} and \eqref{eq:rho_2}, we have
    \begin{align}
        P(x_i|x_{<i})&=R_i(x) + \epsilon\frac{1-2^nR_i(x)}{2^n-1}\\
        Q_i(x_i|x_{<i}) &= R_i(x) + (1-2r)\epsilon\frac{1-2^nR_i(x)}{2^n-1}
    \end{align}
    for all $i=1,\dots,N$ and $x_1,\dots x_i \in \{0,1\}^n$. We also have $P(x_i|x_{<i})=Q(x_i|x_{<i})+2r\epsilon\frac{1-2^nR(x)}{2^n-1}$. With these relations, the KL divergence of $P(x_i|x_{<i})$ and $Q(x_i|x_{<i})$ is bounded as
    \begin{align}
        D_{\rm KL}(P_i(\cdot|x_{<i})\|Q_i(\cdot|x_{<i}))
        &= \sum_{x_i \in\{0,1\}^n}P_i(x_i|x_{<i})\log\left(\frac{P_i(x_i|x_{<i})}{Q_i(x_i|x_{<i})}\right)\\
        &= \sum_{x_i \in\{0,1\}^n}\left(Q_i(x_i|x_{<i}) + 2r\epsilon\frac{1-2^nR_i(x_i|x_{<i})}{2^n-1}\right)
        \log\left(1 + \frac{2r\epsilon}{Q_i(x_i|x_{<i})}\frac{1-2^nR_i(x_i|x_{<i})}{2^n-1}\right)\\
        &\le\sum_{x_i \in\{0,1\}^n}\left(Q_i(x_i|x_{<i}) + 2r\epsilon\frac{1-2^nR_i(x_i|x_{<i})}{2^n-1}\right)\cdot \frac{2r\epsilon}{\ln 2\cdot Q_i(x_i|x_{<i})}\frac{1-2^nR_i(x_i|x_{<i})}{2^n-1}\\
        &=\frac{2r\epsilon}{\ln 2}\sum_{x_i \in\{0,1\}^n}\frac{1-2^nR_i(x_i|x_{<i})}{2^n-1} + \frac{4r^2\epsilon^2}{\ln 2}\sum_{x_i \in\{0,1\}^n}\frac{1}{Q_i(x_i|x_{<i})}\left(\frac{1-2^nR_i(x_i|x_{<i})}{2^n-1}\right)^2\\
        &=\frac{4r^2\epsilon^2}{\ln 2}\sum_{x_i \in\{0,1\}^n}\frac{R^2_i(x_i)}{Q_i(x_i|x_{<i})}\left(\frac{1/R_i(x_i|x_{<i})-2^n}{2^n-1}\right)^2.
    \end{align}
    Since $R_i(x_i|x_{<i})/Q_i(x_i|x_{<i})=1+O(\epsilon)$, we have
    \begin{align}
        D_{\rm KL}(P_i(\cdot|x_{<i})\|Q_i(\cdot|x_{<i}))&\le \frac{4r^2\epsilon^2(1+O(\epsilon))}{\ln 2}\sum_{x_i \in\{0,1\}^n}R_i(x_i|x_{<i})\left(\frac{1/R_i(x_i|x_{<i})-2^n}{2^n-1}\right)^2\\
        &= \frac{4r^2\epsilon^2(1+O(\epsilon))}{\ln 2}\underset{x_i \sim R_i}{\mathbb{E}}\left[\left(\frac{1-1/(2^nR_i(x_i|x_{<i}))}{1-1/2^n}\right)^2\right].
    \end{align}
    As $\ket{\psi}$ has no orthogonal stabilizer state, $R_i(x_i|x_{<i}) = \Theta(1)$ for all $x \in \{0,1\}^n$, and thus,
    \begin{equation}
        \underset{x_i \sim P_i(\cdot|x_{<i})}{\mathbb{E}}\left[\left(\frac{1-1/(2^nP_i(x_i|x_{<i}))}{1-1/2^n}\right)^2\right]=O(1).
    \end{equation}
    Therefore, $D_{\rm KL}(Q_i(\cdot|x_{<i})\|P_i(\cdot|x_{<i})) = O(r^2\epsilon^2)$. By the chain rule of the KL divergence,
    \begin{align}
         D_{\rm KL}(Q\|P)
         &= \sum_{i=1}^N Q_{<i}(x_{<i})D_{\rm KL}(Q_i(\cdot|x_{<i})\|P_i(\cdot|x_{<i}))\\
         &= \sum_{i=1}^N Q_{<i}(x_{<i})O(r^2\epsilon^2)\\
         &= N\cdot O(r^2\epsilon^2),
    \end{align}
    where $Q_{<i}(x_{<i}) = Q(X_1=x_1,\dots,X_{i-1}=x_{i-1})$. Therefore, with Eq.~\eqref{eq:pinsker}, we have
    \begin{equation}
        N \cdot O(r^2\epsilon^2) \ge \frac{2}{9\ln 2},
    \end{equation}
    and it leads to $N=\Omega(1/r^2\epsilon^2)$.

    Finally, since $N=\max\{N_{\epsilon,r},N_{(1-2r)\epsilon,r}\}$, $N_{\epsilon,r} = \Omega(1/r^2\epsilon)$ or $N_{(1-2r)\epsilon,r} =\Omega(1/r^2\epsilon)$ (or both). If $N_{\epsilon,r} = \Omega(1/r^2\epsilon)$, we achieve the desired immediately. If $N_{\epsilon,r} = \Omega(1/r^2\epsilon)$, simply putting $\epsilon'=(1-2r)\epsilon$ yields the same result since $N_{\epsilon',r}=\Omega(1/r^2\epsilon)=\Omega(1/r^2\epsilon')$.
\end{proof}

\begingroup
  \renewcommand{\thetheorem}{\ref{thm:suff_two-copy}}
  \begin{theorem}[restated]
    \label{thm:suff_two-copy-restated}
    For a $n$-qubit state $\ket{\psi} \in \mathcal{H}$, denote $\mathcal{H}=\mathcal{H}_0\oplus \dots \oplus \mathcal{H}_k$, where $\mathcal{H}_j$ are irreducible representations of $G_\psi$ and $\mathcal{H}_0={\rm span}\{\ket{\psi}\}$. If $\mathcal{H}_{0} \not\simeq \mathcal{H}_j$ for all $j=1,\dots,k$, then for $0<r<1$ and for any state $\rho$ with infidelity $\epsilon$, Bell measurement scheme outputs an estimator $\hat{\epsilon}$ satisfying $|\epsilon-\hat{\epsilon}|\le (r+O(\epsilon))\epsilon$ with probability at least $1-\delta$, using $N=O(\log(1/\delta)\cdot1/r^2\epsilon)$ copies of $\rho$.
  \end{theorem}
\endgroup

\begin{proof}
    Given $\rho$, we first apply twirling, i.e.,
    \begin{equation}
        \rho \rightarrow \rho' = \frac{1}{|G_\psi|}\sum_{U\in G_\psi}U\rho U^{\dagger}.
    \end{equation}
    As seen in Eq.~(11) in the main text, $\mathcal{H}_0\not\simeq H_j$ for all $j=1,\dots, k$ implies $\rho'=(1-\epsilon)\ketbra{\psi}{\psi} \oplus \epsilon\sigma$ for some density matrix $\sigma$ of the subspace $\mathcal{H}_1 \oplus \dots \oplus \mathcal{H}_k$.
    Given $\rho' \otimes \rho'$, we perform a Bell measurement on each pair of qubits. For a measurement outcome $x\in\{0,1\}^{2n}$, we set a random variable $X$ as
    \begin{equation}
        X = 
        \begin{cases}
            1,\quad &\text{if $\sum_{i=1}^n x_i \cdot x_{i+n}$ is odd},\\
            0,\quad &\text{otherwise}.
        \end{cases}
    \end{equation}
    As shown in Appendix~\ref{sec:swap_test}, we have
    Then $\mathbb{P}[X=1]=(1-\Tr(\rho'^2))/2$. Therefore, $\mathbb{P}[X=1]=\epsilon-\epsilon^2(1+\Tr(\sigma^2))/2=\epsilon(1+O(\epsilon))$, and thus $X \sim {\rm Bern}(\epsilon(1+O(\epsilon)))$
    By repeating it $N$ times, we have samples $X_1,\dots ,X_N$. We set the estimator
    \begin{equation}
        \hat{\epsilon}=\frac{1}{N}\sum_{j=1}^N X_j.
    \end{equation}
    It is straightforward that $\mathbb{E}[\hat{\epsilon}]=(1+O(\epsilon))\epsilon$, and by Lemma~\ref{lem:chernoff}, we have
    \begin{equation}
        \mathbb{P}[|\hat{\epsilon}-\mathbb{E}[\hat{\epsilon}]|\ge r\mathbb{E}[\hat{\epsilon}]] \le 2\exp\left(-\frac{r^2\mathbb{E}[\hat{\epsilon}]}{3}N\right),
    \end{equation}
    for $0<r<1$. By the triangle inequality,
    \begin{align}
        |\epsilon-\hat{\epsilon}| &\le |\epsilon-\mathbb{E}[\hat{\epsilon}]|+|\hat{\epsilon}-\mathbb{E}[\hat{\epsilon}]|\\
        &= r\mathbb{E}[\hat{\epsilon}]+ O(\epsilon^2)\\
        &= (r+O(\epsilon))\epsilon,
    \end{align}
    with probability at least $1-2\exp\left(-\frac{r^2\mathbb{E}[\hat{\epsilon}]}{3}N\right)$. Therefore, we achieve
    \begin{equation}
        |\epsilon-\hat{\epsilon}|\le (r+O(\epsilon))\epsilon
    \end{equation}
    with probability at least $1-\delta$ by setting $N=\ceil{\frac{3}{r^2\mathbb{E}[\hat{\epsilon}]}(\ln1/\delta+\ln2)}=O(\log(1/\delta)\cdot1/r^2\epsilon)$.
\end{proof}

\begingroup
  \renewcommand{\thetheorem}{\ref{thm:suff_single-copy}}
  \begin{theorem}[restated]
    \label{suff_single-copy-restated}
    For a $n$-qubit state $\ket{\psi} \in \mathcal{H}$, with $n$ being constant larger than 1, denote $\mathcal{H}=\mathcal{H}_0\oplus \mathcal{H}_1\oplus \dots \oplus \mathcal{H}_k$, where $\mathcal{H}_j$ are irreducible representations of $G_\psi$ and $\mathcal{H}_0={\rm span}\{\ket{\psi}\}$. Suppose there exist stabilizer states $\ket{s_1}\in\mathcal{H}_1,\dots,\ket{s_k}\in\mathcal{H}_k$. Then, there is a single-copy scheme that, for $0<r<1$ and any state $\rho$ with infidelity $\epsilon$, outputs an estimator $\hat{\epsilon}$ satisfying $|\epsilon-\hat{\epsilon}|\le r\epsilon$ with probability at least $1-\delta$, using $N=O(\log(1/\delta)\cdot1/r^2\epsilon)$ copies of $\rho$.
  \end{theorem}
\endgroup

\begin{proof}
    Suppose $\ket{\psi}$ is an $n\ge2$-qubit state, where $n$ is constant. Since $\rho'_{ii} = \lambda_i\Pi_i$ for $i=1,\dots,k$, we have
    \begin{equation}
        \epsilon=\sum_{i=1}^k \lambda_i \cdot \dim\mathcal{H}_i.
    \end{equation}
    Our strategy is to estimate each $\lambda_i$ using the stabilizer state $\ket{s_i}$ for $i=1,\dots,k$. Specifically, for each $i=1,\dots, k$, we choose $U_i \in \mathcal{C}_n$ such that $U_i^\dagger \ket{0^n}=\ket{s_i}$. We apply $U_i$ to the twirled state $\rho'$ and measure all qubits in the computational basis state. Then, for a measurement outcome $x \in \{0,1\}^n$, we set a random variable $X^{(i)}$ as
    \begin{equation}
        X^{(i)} = 
        \begin{cases}
            1,\quad &\text{if $x=0^n$},\\
            0,\quad &\text{otherwise}.
        \end{cases}
    \end{equation}
    Since $\langle s_i|\rho'|s_i \rangle = \langle s_i|\rho'_{ii}|s_i \rangle=\lambda_i$, $X^{(i)}\sim {\rm Bern}(\lambda_i)$. Repeating it $N_i$ times, we have samples $X^{(i)}_1,\dots,X^{(i)}_{N_i}$. We set the estimator of $\lambda_i$ as
    \begin{equation}
        \hat{\lambda}_i = \frac{1}{N}\sum_{j=1}^N X^{(i)}_j,
    \end{equation}
    then it is straightforward that $\mathbb{E}[\hat{\lambda}_i]=\lambda_i$. By Lemma~\ref{lem:chernoff_2},
    \begin{align}
        \mathbb{P}\left[\hat{\lambda}_i-\lambda_i \ge \frac{r\epsilon}{k\cdot \dim\mathcal{H}_i}\right]
        &< \exp\left(-\frac{\left(\frac{r\epsilon}{k\dim\mathcal{H}_i\lambda_i}\right)^2}{2+\frac{r\epsilon}{k\dim\mathcal{H}_i\lambda_i}}N\lambda_i\right)
        \\
        &= \exp\left(-\frac{\left(\frac{r}{k\dim\mathcal{H}_i}\right)^2}{2\lambda_i+\frac{r\epsilon}{k\dim\mathcal{H}_i}}\epsilon^2 N\right)\\
        &\le \exp\left(-\left(\frac{r}{k\dim\mathcal{H}_i}\right)^2 \frac{\epsilon N}{3}\right)
        \label{eq:temp}
    \end{align}
    Meanwhile, by applying Lemma~\ref{lem:chernoff}, we have
    \begin{align}
        \mathbb{P}\left[\hat{\lambda}_i-\lambda_i \le -\frac{r\epsilon}{k\cdot \dim\mathcal{H}_i}\right]
        &< \exp\left(-\left(\frac{r\epsilon}{k\dim\mathcal{H}_i\lambda_i}\right)^2\frac{N\lambda_i}{2}\right)\\
        &\le \exp\left(-\left(\frac{r}{k\dim\mathcal{H}_i}\right)^2\frac{\epsilon N}{2}\right).
        \label{eq:temp____}
    \end{align}
    By combining Eqs.~\eqref{eq:temp} and \eqref{eq:temp____}, we have
    \begin{equation}
        \mathbb{P}\left[\left\vert\hat{\lambda}_i-\lambda_i \right\vert\le \frac{r\epsilon}{k\cdot \dim\mathcal{H}_i}\right] < 2\exp\left({-\left(\frac{r}{k\dim\mathcal{H}_i}\right)^2\frac{\epsilon N}{3}}\right).
    \end{equation}
    By choosing $N_i = \left\lceil \frac{3k\dim\mathcal{H}_i}{r^2 \epsilon}(\ln k/\delta + \ln2) \right\rceil$, we have $\left\vert\hat{\lambda}_i-\lambda_i\right\vert \le \frac{r\epsilon}{k\cdot \dim\mathcal{H}_i}$ with probability at least $\delta/k$.

    Finally, we set
    \begin{equation}
        \hat{\epsilon}=\sum_{i=1}^{k} \hat{\lambda}_i\cdot \dim\mathcal{H}_i.
    \end{equation}
    Then it leads to
    \begin{align}
        \left\vert\hat{\epsilon} -\epsilon\right\vert
        &\le \sum_{i=1}^k\left\vert \hat{\lambda}_i -\lambda_i\right\vert \cdot \dim \mathcal{H}_i\\
        &\le \sum_{i=1}^{k} \frac{r\epsilon}{k\dim\mathcal{H}_i} \cdot \dim\mathcal{H}_i\\
        &=r\epsilon,
    \end{align}
    with probability at least $(1-\delta/k)^k \ge 1-\delta$. The corresponding sample complexity is
    \begin{align}
        N
        &= \sum_{i=1}^k N_i\\
        &= \sum_{i=1}^k \left\lceil \frac{3k\dim\mathcal{H}_i}{r^2 \epsilon}(\ln k/\delta + \ln2) \right\rceil\\
        &= O(\log(1/\delta)\cdot1/r^2\epsilon),
    \end{align}
    since both $k$ and $\dim\mathcal{H}_i$'s are constant.
\end{proof}

\begingroup
  \renewcommand{\thetheorem}{\ref{thm:optimality}}
  \begin{theorem}[restated]
    \label{thm:optimality-restated}
    Let $\ket{\psi}$ be an $n$-qubit state. Suppose there exists a benchmarking scheme that allows arbitrary joint measurements on multiple copies of $\rho$ and, for any noisy state $\rho$ with infidelity $\epsilon$, outputs an estimator $\hat{\epsilon}$ satisfying $|\hat{\epsilon} - \epsilon| \leq r\epsilon$ with probability at least $2/3$, while consuming $N$ copies of $\rho$ in total. Then, $N = \Omega(1/r^2\epsilon)$.
  \end{theorem}
\endgroup

\begin{proof}
    Suppose $\ket{\psi}$ is an $n$-qubit state, and let $\hat{\epsilon}=\hat{\epsilon}(\rho^{\otimes N})$ the estimator that outputs infidelity of $\rho$ by performing a positive operator-valued measurement (POVM) on $\rho^{\otimes N}$. To be concrete, let a set $\{(\epsilon_j,\Lambda_j)\}_j$ be the POVM where $\Lambda_j$'s are positive semi-definite operators such that $\sum_j\Lambda_j=I$, and $\hat{\epsilon}_j$ be the corresponding estimation values. With this notation, $\hat{\epsilon}(\rho^{\otimes N})=\epsilon_j$ with probability $\Tr\left(\rho^{\otimes N}\Lambda_j\right)$. Let us further assume that given the input state's infidelity $\epsilon$ and the multiplicative precision $r$, $\hat{\epsilon}$ guarantees to be $|\hat{\epsilon} - \epsilon|\le r\epsilon$ with probability at least $2/3$, if $N \ge N_{\epsilon,r}$.
    
    We first show that using this estimator $\hat{\epsilon}$, we can distinguish the following states:
    \begin{align}
        \rho_{0} &= (1-\epsilon)\ketbra{\psi}{\psi}+\epsilon\ketbra{\psi^\perp}{\psi^\perp},\\
        \rho_{1} &= (1-(1-2r)\epsilon)\ketbra{\psi}{\psi} + (1-2r)\epsilon\ketbra{\psi^\perp}{\psi^\perp},
        \label{eq:rho1}
    \end{align}
    with probability at least 2/3. Here, we take $r<1/2$, and $\ket{\psi^\perp}$ is a state that is orthogonal to $\ket{\psi}$. Specifically, note that $\rho_0$ and $\rho_1$ have infidelities $\epsilon$ and $(1-2r)\epsilon$, respectively. Therefore, choosing $N =\max\{N_{\epsilon,r},N_{(1-2r)\epsilon, r}\}$, we have
    \begin{align}
        \hat{\epsilon}(\rho_0^{\otimes N}) &\ge (1-r)\epsilon & \text{with probability $\ge 2/3$},\\
        \hat{\epsilon}(\rho_1^{\otimes N}) &\le (1+r)(1-2r)\epsilon<(1-r)\epsilon & \text{with probability $\ge 2/3$}.
    \end{align}
    With this relation, we can construct another POVM $\{\Lambda'_0,\Lambda'_1\}$ that distinguishes $\rho_0$ and $\rho_1$, where
    \begin{align}
        \Lambda'_0 &= \sum_{\Lambda_j: \hat{\epsilon}_j \ge (1-r)\epsilon} \Lambda_j\\
        \Lambda'_1 &= \sum_{\Lambda_j: \hat{\epsilon}_j < (1-r)\epsilon} \Lambda_j.
    \end{align}
    As $\Tr\left(\rho_0^N \Lambda_0\right)\ge2/3$ and $\Tr\left(\rho_1^N \Lambda_1\right)\ge2/3$, we can distinguish $\rho_0$ and $\rho_1$ with probability at least $2/3$.

    Now, we show that $N$ should be $\Omega(1/r^2\epsilon)$ for such POVM to exist. Since $\frac{1}{2}\left\|\rho-\sigma\right\|_1 = \max_{0\le \Lambda\le I}\Tr\left((\rho-\sigma)\Lambda\right)$ for any density matrices $\rho$ and $\sigma$, we have
    \begin{align}
        \left\|\rho_0^{\otimes N}-\rho_1^{\otimes N}\right\|_1
        &\ge 2\Tr\left((\rho_0^{\otimes N}-\rho_1^{\otimes N})\Lambda'_0\right)\\
        &\ge 2\times(2/3-1/3)\\
        &=2/3,
    \end{align}
    where we used $\Lambda'_0+\Lambda'_1=I$ for the second inequality. By the quantum Pinsker inequality and the fact that the quantum relative entropy is tensorized (see e.g., Ref.~\cite{Wilde_2013}),
    \begin{align}
        \frac{1}{2\ln 2}\|\rho_0^{\otimes N}-\rho_1^{\otimes N}\|_1^2
        &\le S(\rho_0^{\otimes N}\|\rho_1^{\otimes N})\\
        &=N\cdot S(\rho_0\|\rho_1),
    \end{align}
    where $S(\rho\|\sigma)$ is relative entropy of some states $\rho$ and $\sigma$. Meanwhile,
    \begin{align}
        S(\rho_0\|\rho_1)
        &=\Tr\rho_0 (\log\rho_0 - \log\rho_1) \\
        &= (1-\epsilon)\log\left(\frac{1-\epsilon}{1-(1-2r)\epsilon}\right)+\epsilon\log\left(\frac{1}{1-2r}\right)\\
        &=(1-\epsilon)\log\left(1-\frac{2r\epsilon}{1-(1-2r)\epsilon}\right)-\epsilon\log(1-2r)\\
        &\le -\frac{2r\epsilon}{\ln2}\cdot\frac{1-\epsilon}{1-(1-2r)\epsilon} + \frac{2r\epsilon}{\ln 2}(1+r+O(r^2))\\
        &\le \frac{2r\epsilon}{\ln2}\left(-1+\frac{2r\epsilon}{1-(1-2r)\epsilon}+1+r+O(r^2)\right)\\
        &\le \frac{2r\epsilon}{\ln2}\left(r+O(r\epsilon)+O(r^2)\right)\\
        &=O(r^2\epsilon).
    \end{align}
    Therefore, $N \cdot O(r^2\epsilon) \ge \frac{1}{2\ln 2}\|\rho_0^{\otimes N}-\rho_1^{\otimes N}\|_1^2$. Since $\frac{1}{2\ln 2}\|\rho_0^{\otimes N}-\rho_1^{\otimes N}\|_1^2 \ge 2/3$, we have $N=\Omega(1/r^2\epsilon)$.

    Finally, since $N=\max\{N_{\epsilon,r},N_{(1-2r)\epsilon,r}\}$, $N_{\epsilon,r} = \Omega(1/r^2\epsilon)$ or $N_{(1-2r)\epsilon,r} =\Omega(1/r^2\epsilon)$ (or both). If $N_{\epsilon,r} = \Omega(1/r^2\epsilon)$, we are done immediately. If $N_{\epsilon,r} = \Omega(1/r^2\epsilon)$, simply putting $\epsilon'=(1-2r)\epsilon$ yields the same result since $N_{\epsilon',r}=\Omega(1/r^2\epsilon)=\Omega(1/r^2\epsilon')$.
\end{proof}

\setcounter{theorem}{5}

\section{\label{sec:twirling}Detailed analysis of state twirling}

In this section, we provide a detailed explanation of state twirling.
We begin with a Pauli framework we developed for analyzing state twirling operations with respect to subgroups of the Clifford group. This framework turns out to be useful for our case studies for multi-qubit magic states. Specifically, we subsequently illustrate the specific consequences of state twirling protocols applied to the $\ket{\text{CZ}}$ and $\ket{\text{CCZ}}$ magic states. In addition, we also present an explicit algorithm that tells whether the conditions for Bell measurement scheme and single-copy scheme are satisfied for an arbitrary magic state.

\subsection{Pauli framework for state twirling with Clifford subgroup}

We first examine the general effect of twirling with a subgroup $G$ of the $n$-qubit Clifford group $\mathcal{C}_n$. For an arbitrary $n$-qubit density matrix $\rho$, the twirling operation with respect to the subgroup $G$ is defined as
\begin{equation}
    \rho \mapsto \rho' = \frac{1}{|G|} \sum_{U\in G} U \rho U^\dagger.
\end{equation}
To understand the effect of this operation, we investigate how the subgroup $G$ partitions the set of Hermitian Pauli operators. Explicitly, the set of Hermitian Pauli operators $\mathcal{P}^{\text{Herm}}_n$ is
\begin{equation}
\mathcal{P}^{\text{Herm}}_n= \{\pm P_1 \otimes \dots \otimes P_n : P_j \in \{I,X,Y,Z\}, \quad j=1,\dots,n\}.
\end{equation}
We introduce an equivalence relation $\sim_G$ on the set $\mathcal{P}^{\text{Herm}}_n$ defined by
\begin{equation}
    P \sim_G P' \quad \text{if and only if} \quad \exists U \in G \text{ such that } U P U^\dagger = P',
\end{equation}
for $P, P' \in \mathcal{P}^{\rm Herm}_n$. Since $G$ is a subgroup of the Clifford group, it is straightforward to verify that $\sim_G$ is indeed an equivalence relation, thus partitioning $\mathcal{P}^{\rm Herm}_n$ into distinct equivalence classes. Note also that letting $K \subset \mathcal{P}^{\rm  Herm}_n$ be an equivalence class, $UPU^\dagger \in K$ for all $P \in K$ and $U \in G$, as $U$ is a Clifford gate. Since every $U\in G$ is invertible, $U$ only permutes the elements in $K$, i.e., $UKU^\dagger:=\{UPU^\dagger:P\in K\} = K$.

With this partition into equivalence classes, we can simplify the twirled state $\rho'$ as follows. Since the density matrix $\rho$ can be expressed as a linear combination of the elements in $\mathcal{P}^{\rm Herm}_n$, we can denote
\begin{equation}
    \rho = \sum_{P\in\mathcal{P}^{\rm Herm}_n} c_P P,
\end{equation}
where $c_P$'s are real numbers. We decompose $\mathcal{P}^{\rm Herm}_n$ as
\begin{equation}
    \mathcal{P}^{\rm Herm}_n=\bigsqcup_{\alpha}K_\alpha,
\end{equation}
where $K_\alpha$'s are the distinct equivalence classes with respect to $\sim_G$. This leads to
\begin{align}
    \rho' &= \frac{1}{|G|}\sum_{\alpha} \left(\sum_{P\in K_\alpha} c_P\sum_{U\in G}U P U^{\dagger}\right).
    \label{eq:rho'_twirling}
\end{align}
Meanwhile, for all $P,P' \in K_\alpha$, there exists $V \in G$ such that $P' = VPV^\dagger$, and thus
\begin{align}
    \sum_{U\in G} UP'U^\dagger = \sum_{U\in G} (UV) P(UV)^\dagger = \sum_{U \in G}UPU^\dagger.
\end{align}
Therefore, 
\begin{align}
    \frac{1}{|G|}\sum_{U\in G} U P U^\dagger = \frac{1}{|G| |K_\alpha|}\sum_{U\in G}\sum_{P \in K_\alpha} UPU^\dagger = \frac{1}{|K_\alpha|}\sum_{P\in K_\alpha} P.
\end{align}
Plug this in Eq.~\eqref{eq:rho'_twirling}, we have
\begin{align}
    \rho' &= \sum_{\alpha} \frac{\sum_{P\in K_\alpha} c_P}{|K_\alpha|}\cdot \left(\sum_{P\in K_\alpha}P\right).
\end{align}
By re-parametrizing it with $C_\alpha=\frac{1}{|K_\alpha|}\sum_{P\in K_\alpha} c_P$,
\begin{equation}
    \rho' = \sum_{\alpha} C_\alpha \cdot \left(\sum_{P\in K_\alpha}P\right).
\end{equation}
Consequently, the twirling process collapses all the coefficients $\{c_P\}_{P\in K_\alpha}$ within the same equivalent class $K_\alpha$ into a single parameter $C_\alpha$, reducing the number of parameters of the density matrix.

\subsection{Case studies for twirling magic states}

Now we focus on the specific magic states including $\ket{T}$, $\ket{H}$, $\ket{\rm CZ}$, $\ket{\rm CCZ}$ states,
\begin{align}
    \ketbra{T}{T} &= \frac{1}{2}\left(I + \frac{1}{\sqrt{3}}(X+Y+Z)\right),\\
    \ketbra{H}{H} &= \frac{1}{2}\left(I + \frac{1}{\sqrt{2}}(X+Y)\right),\\
    \ket{\rm CZ} &= \frac{1}{\sqrt{3}}(\ket{00}+\ket{01}+\ket{10}),\\
    \ket{\rm CCZ} &= \frac{1}{2}(\ket{00+}+\ket{01+}+\ket{10+}+\ket{11-}).
\end{align}
and analyze state twirling by applying the above framework.

\subsubsection{Single-qubit magic states}
To begin with, we discuss the twirling of single qubit magic states $\ket{T}$ and $\ket{H}$. First, it is straightforward to observe that $G_T = \langle U_0 \rangle$, where
\begin{equation}
    U_0 = \frac{e^{i\pi/4}}{\sqrt{2}}
    \begin{pmatrix}
        1 & 1 \\
        i & -i
    \end{pmatrix}.
\end{equation}
As shown above, $G_T$ partitions $P_1^{\rm Herm}$ into four distinct equivalent classes:
\begin{equation}
\begin{split}
    &K_{0,+}^{T} = \{I\}, \quad K_{1,+}^{T} = \{X, Y, Z\};\\
    &K_{0,-}^{T} = \{-I\}, \quad K_{1,+}^{T} = \{-X, -Y, -Z\}.
\end{split}
\end{equation}
Therefore,
\begin{equation}
    \rho'_T = C_{0,+}\left(\sum_{P\in K_{0,+}^{T}}P\right) + C_{1,+}\left(\sum_{P\in K_{1,+}^{T}}P\right) + C_{0,-}\left(\sum_{P\in K_{0,-}^{T}}P\right) + C_{1,-}\left(\sum_{P\in K_{1,-}^{T}}P\right).
\end{equation}
Here, it is straightforward that $\sum_{P \in K_{m,+}^{T}}P = -\sum_{P \in K_{m,-}^{T}}P$ for $m=0,1$. Therefore, by denoting $C_0 = C_{0,+}-C_{0,-}=1/2$ and $C_1 = C_{1,+} - C_{1,-}$, we have
\begin{equation}
    \rho'_T = \frac{1}{2} + C_1\sum_{P\in K_{1,+}}P.
\end{equation}
Taking $C_1 = \frac{1-2\epsilon}{2\sqrt{3}}$ leads to
\begin{equation}
    \rho'_T = (1-\epsilon) \ketbra{T}{T} + \epsilon\ketbra{T^\perp}{T^\perp}
\end{equation}
for a state $\ket{T^\perp}$ orthogonal to $\ket{T}$.

Similarly, observing $G_H = \langle S \rangle$ with $S={\rm diag}(1, i)$ partitions $P_1^{\rm Herm}$ into five distinct equivalent classes:
\begin{equation}
\begin{split}
    &K_{0,+}^{H} = \{I\}, \quad K_{1,+}^{H} = \{X, Y\};\\
    &K_{0,-}^{H} = \{-I\}, \quad K_{1,-}^{H} = \{-X, -Y\};\\
    &K_{2}^{H} = \{Z,-Z\}.
\end{split}
\end{equation}
Again, noticing that $\sum_{P \in K_{m,+}^{T}}P = -\sum_{P \in K_{m,-}^{T}}P$ for $m=0,1$ and $\sum_{P \in K_{2}}P = 0$ leads to
\begin{equation}
    \rho'_H = \frac{1}{2} + C_1\sum_{P\in K_{1,+}}P.
\end{equation}
With the parametrization of $C_1 = \frac{1-2\epsilon}{2\sqrt{2}}$, we have
\begin{equation}
    \rho'_H = (1-\epsilon)\ketbra{H}{H} + \epsilon\ketbra{H^\perp}{H^\perp},
\end{equation}
for a state $\ket{H^\perp}$ orthogonal to $\ket{H}$.

\subsubsection{$\ket{\rm CZ}$ state}

Now we discuss the two-qubit magic state $\ket{\rm CZ}$, and the corresponding twirling group $G_{\rm CZ} = \langle {\rm CZ}, {\rm SWAP}, X_2 \cdot {\rm CX}_{1,2}\rangle$ where $\ket{\rm CZ}$ is an eigenstate of every element in $G_{\rm CZ}$. $G_{\rm CZ}$ partitions $\mathcal{P}^{\rm Herm}_{2}$ into seven distinct equivalent classes:
\begin{equation}
\begin{split}
    & K_{0,+}^{\rm CZ} = \{II\}, \quad K_{1,+}^{\rm CZ} = \{IZ,ZI,-ZZ\}, \quad K_{2,+}^{\rm CZ} = \{IX,XI,XZ,ZX,XX,YY\};\\
    & K_{m,-}^{\rm CZ} = \{-P|P\in K_{m,+}^{\rm CZ}\} \quad (m=0,1,2);\\
    & K_3^{\rm CZ} = \{\pm YI, \pm IY, \pm XY, \pm YX, \pm YZ, \pm ZY\}.
\end{split}
\end{equation}
It is straightforward that $\sum_{P\in K_{m,+}^{\rm CZ}}P = -\sum_{P\in K_{m,-}^{\rm CZ}}P$ for $m = 0,1,2$, and $\sum_{P\in K_3^{\rm CZ}}P = 0$. Therefore, for an arbitrary $2$-qubit density matrix $\rho$, the twirled state $\rho'$ is in the form of
\begin{equation}
    \rho'_{\rm CZ} = \frac{1}{4}I + C_1 \sum_{P\in K_{1,+}^{\rm CZ}}P + C_2 \sum_{P\in K_{2,+}^{\rm CZ}}P.
\end{equation}
After a further re-parametrization with $C_1 = \frac{1}{12} - \frac{\epsilon_1}{3}$, $C_2 = \frac{1}{6} - \frac{\epsilon_1}{6} - \frac{\epsilon_2}{4}$, we have
\begin{equation}
    \rho'_{\rm CZ} = (1-\epsilon_1 -\epsilon_2)\ketbra{\rm CZ}{\rm CZ}+\epsilon_1\Pi_1 + \frac{\epsilon_2}{2}\Pi_2,
\end{equation}
where $\Pi_1=\ketbra{11}{11}$ and $\Pi_2=I-\ketbra{\rm CZ}{\rm CZ} -\ketbra{11}{11}$.

\subsubsection{$\ket{\rm CCZ}$ state}

Finally, we discuss the twirling of the three-qubit magic state $\ket{\rm CCZ}$. It is convenient to notice that
\begin{equation}
    \ket{\rm CCZ} = {\rm CCZ}\ket{+++},
\end{equation}
where the controlled-controlled-Z (CCZ) gate is defined as ${\rm CCZ} = {\rm diag}(1,1,1,1,1,1,1,-1)$ in the computational basis. Let us consider state twirling of $\ket{+++}$ instead, and use this result for twirling $\ket{\rm CCZ}$ later. Note that the twirling group of $\ket{+++}$ is given as
\begin{equation}
    G_{+++} = \langle X_i, {\rm CNOT}_{j,k} \rangle \quad\text{(for $i,j,k=1,2,3$)}.
\end{equation}
This twirling group partitions $\mathcal{P}^{\rm Herm}_3$ into distinct equivalent classes including:
\begin{equation}
\begin{split}
    & K_{0,+}^{+++} = \{III\}, \quad K_{1,+}^{+++} = \{X_1, X_2, X_3, X_1 X_2, X_1 X_3, X_2 X_3, X_1 X_2 X_3\};\\
    & K_{m,-}^{\rm +++} = \{-P|P\in K_{m,+}^{\rm CZ}\} \quad (m=0,1).
\end{split}
\end{equation}
There are more distinct classes involving Pauli operators acting $Y$ or $Z$ on any qubit. However, each of these classes, say $K$, should satisfy $K=-K:=\{-P:P\in K\}$, and thus all are canceled out in the expression of the twirled state $\rho'_{+++}$. Therefore, the twirled state is in the form of
\begin{equation}
    \rho_{+++}' = \frac{1}{8}I + C_1\sum_{P\in K_{1,+}^{XXX}} P,
\end{equation}
and with the re-parametrization of $C_1 = \frac{1}{8} - \frac{\epsilon}{7}$,
\begin{equation}
    \rho_{+++}' = (1-\epsilon)\ketbra{+++}{+++} + \epsilon\frac{I-\ketbra{+++}{+++}}{7}.
    \label{eq:+++state}
\end{equation}

Now, we use this result to analyze the twirling with respect to $\ket{\rm CCZ}$. To this end, we first show that for all $U \in G_{+++}$, ${\rm CCZ} \cdot U \cdot {\rm CCZ}$ are Clifford operations. This is directly shown from that facts that $G_{+++}$ is generated by $X$ and ${\rm CNOT}$ gates, and
\begin{align}
    &{\rm CCZ} \cdot X_1 \cdot {\rm CCZ} = X_1 \cdot {\rm CZ}_{2,3}\\
    &{\rm CCZ} \cdot {\rm CNOT}_{2,3} \cdot {\rm CCZ} = {\rm CZ}_{1,2} \cdot {\rm CNOT}_{2,3}.
\end{align}
In addition,
\begin{equation}
    {\rm CCZ} \cdot U \cdot {\rm CCZ} \ket{\rm CCZ} \propto \ket{\rm CCZ}.
\end{equation}
for all $U \in G_{+++}$. Therefore, the twirling group for $\ket{\rm CCZ}$ state is
\begin{equation}
    G_{\rm CCZ} = \{{\rm CCZ} \cdot U \cdot {\rm CCZ}: U\in G_{+++}\},
\end{equation}
Therefore, using Eq.~\eqref{eq:+++state}, the twirled state $\rho'_{\rm CCZ}$ is in the form of
\begin{align}
    \rho'_{\rm CCZ}
    &= {\rm CCZ}\left( \sum_{U \in G_{+++}}U ({\rm CCZ}\rho {\rm CCZ})U^\dagger \right) {\rm CCZ}\\
    &= {\rm CCZ}\left( (1-\epsilon)\ketbra{+++}{+++} + \epsilon\frac{I-\ketbra{+++}{+++}}{7} \right) {\rm CCZ}\\
    &= (1-\epsilon)\ketbra{\rm CCZ}{\rm CCZ} + \frac{\epsilon}{7}\Pi_1,
\end{align}
where $\Pi_1 = I-\ketbra{\rm CCZ}{\rm CCZ}$.

\subsection{Algorithms for applicability of benchmarking schemes}

While the Pauli framework introduced above provides a powerful way to describe the effect of state twirling for the well-known magic states analyzied above, it may not be straightforward for an arbitrary magic state. Here, we present an explicit algorithm that given an $n$-qubit state $\ket{\psi}$, tells whether the conditions for Bell measurement scheme (Theorem~\ref{thm:suff_two-copy-restated}) or single-copy scheme (Theorem~\ref{thm:suff_single-copy}) are satisfied.

\begin{figure}
\begin{algorithm}[H] 
\caption{Check applicability of the Bell measurement scheme}
\label{alg:bell_measurement}
\begin{algorithmic}[1]
    \renewcommand{\algorithmicrequire}{\textbf{Input:}}
	\renewcommand{\algorithmicensure}{\textbf{Output:}}
    \Require State $\ket{\psi} \in \mathcal{H}$, corresponding twirling group $G_\psi$, and its irreducible decomposition 
    $\mathcal{H} = \bigoplus_{j=0}^k \mathcal{H}_j$ with $\mathcal{H}_0 = \mathrm{span}\{\ket{\psi}\}$.
    \Ensure Applicability of the Bell measurement scheme (\textbf{Yes} or \textbf{No}).
    \vspace{0.5em}
    \For{$j = 1, \dots, k$}
        \If{$\dim(\mathcal{H}_j) = 1$}
            \State Let $\mathcal{H}_j = \mathrm{span}\{\ket{\phi}\}$.
            \State Check whether $U\ket{\phi} = e^{i\theta_U}\ket{\phi}, \quad \forall\, U \in G_\psi$.
            \If{the above condition holds}
                \State \Return \textbf{No}
            \EndIf
        \EndIf
    \EndFor
    \State \Return \textbf{Yes}
\end{algorithmic}
\end{algorithm}
\end{figure}

First, to check the applicability of the Bell measurement scheme, we first decompose the Hilbert space $\mathcal{H}$ into direct sums of irreducible representations of $G_\psi$, i.e.,
\begin{equation}
    \mathcal{H}=\mathcal{H}_0 \oplus \mathcal{H}_1 \oplus \cdots \oplus \mathcal{H}_k,
\end{equation}
where we take $\mathcal{H}_0 = {\rm span}\{\ket{\psi}\}$. Let us assume that the explicit decomposition is given (this can be done using character theory of finite groups, see Chap. 2 of~\cite{fultonRepresentationTheory2004}). By definition of $G_\psi$, there exists a phase $\theta_U \in [0,2\pi)$ such that $U\ket{\psi} = e^{i\theta_U}\ket{\psi}$ for all $U \in G_\psi$.

With this setup, the Bell measurement scheme is applicable if $\mathcal{H}_0 \not\simeq \mathcal{H}_j$ for all $j=1,\dots,k$. Verifying this condition is straightforward: since two irreducible representations with different dimensions cannot be equivalent, it suffices to inspect one-dimensional irreducible representations among $\mathcal{H}_1,\dots,\mathcal{H}_k$. If $\mathcal{H}_j$ is one-dimensional, say $\mathcal{H}_j = {\rm span}\{\ket{\phi}\}$, then $\mathcal{H}_0 \simeq \mathcal{H}_j$ if and only if $U\ket{\phi} = e^{i\theta_U}\ket{\phi}$ for all $U \in G_\psi$. The overall procedure for this verification is summarized in Algorithm~\ref{alg:bell_measurement}.

\begin{figure}
\begin{algorithm}[H] 
\caption{Applicability of the single-copy scheme}
\label{alg:single_copy}
\begin{algorithmic}[1]
    \renewcommand{\algorithmicrequire}{\textbf{Input:}}
	\renewcommand{\algorithmicensure}{\textbf{Output:}}
    \Require State $\ket{\psi} \in \mathcal{H}$, corresponding twirling group $G_\psi$, its irreducible decomposition 
    $\mathcal{H} = \bigoplus_{j=0}^k \mathcal{H}_j$ with $\mathcal{H}_0 = \mathrm{span}\{\ket{\psi}\}$ and the partitions $\bigsqcup_{\alpha} S_\alpha = \{\mathcal H_1,\dots,\mathcal H_k\}$, where each $S_\alpha$ is an equivalence class under $\simeq$.
    \Ensure Applicability of the single-copy scheme (\textbf{Yes} or \textbf{No}).
    \vspace{0.5em}

    \For{each equivalence class $S_\alpha$}
        \If{$|S_\alpha| = 1$}
            \State Check whether the unique irreducible component in $S_\alpha$ contains at least one stabilizer state.
            \If{the above condition holds}
                \State \Return \textbf{No}
            \EndIf
        \Else
            \State Find all stabilizer states in $\mathcal{H}_{S_\alpha} = \bigoplus_{\mathcal{H}_j \in S_\alpha} \mathcal{H}_j$: 
            $\{\ket{s_1}, \dots, \ket{s_m}\}$.
            \For{each $\ket{s_i}$}
                \State Compute its orbit under $G_\psi$: $O_{s_i} = \{\, U\ket{s_i} : U \in G_\psi \,\}$.
                \State Set $\mathcal{V}_{s_i} = \mathrm{span}(O_{s_i})$.
            \EndFor
            \State Check whether there exist $|S_\alpha|$ orbit subspaces 
            $\{\mathcal{V}_{s_{i_1}},\dots,\mathcal{V}_{s_{i_{|S_\alpha|}}}\}$ satisfying
                $\mathcal{H}_{S_\alpha} = \mathcal{V}_{s_{i_1}} \oplus \mathcal{V}_{s_{i_2}} \oplus \cdots \oplus \mathcal{V}_{s_{i_{|S_\alpha|}}}$.
            \If{the above condition holds}
                \State \Return \textbf{No}
            \EndIf
        \EndIf
    \EndFor
    \State \Return \textbf{Yes}
\end{algorithmic}
\end{algorithm}
\end{figure}

Next, we present an algorithm for checking the applicability of the single-copy scheme in Theorem~\ref{thm:suff_single-copy}. As before, we decompose the Hilbert space $\mathcal{H}$ into irreducible representations of $G_\psi$:
\begin{equation}
    \mathcal{H} = \mathcal{H}_0 \oplus \mathcal{H}_1 \oplus \cdots \oplus \mathcal{H}_k,
\end{equation}
where $\mathcal{H}_0 = \mathrm{span}\{\ket{\psi}\}$. The single-copy scheme is applicable if every irreducible component $\mathcal{H}_j$ (for $j=1,\dots,k$) contains at least one stabilizer state. At first glance, one might expect that applicability can be verified simply by checking each irreducible component separately. However, a subtlety arises when some irreducible representations are equivalent: if $\mathcal{H}_i \simeq \mathcal{H}_j$, then the corresponding subspaces can be freely mixed by choosing a different but equivalent decomposition $\mathcal{H}'_i \oplus \mathcal{H}'_j = \mathcal{H}_i \oplus \mathcal{H}_j$. Hence, when equivalent irreps are present, the decomposition of $\mathcal{H}$ is not unique. The single-copy scheme is applicable whenever there exists \emph{some} valid decomposition such that each irreducible representation contains at least one stabilizer state.

To address this issue, we partition the set of irreducible representations $\{\mathcal H_1,\dots,\mathcal H_k\}$ into equivalence classes under $\simeq$:
\begin{equation}
    \{\mathcal H_1,\dots,\mathcal H_k\} = \bigsqcup_{\alpha} S_\alpha,
\end{equation}
where each $S_\alpha$ collects equivalent irreducible representations. For each class $S_\alpha$, we define the aggregate subspace
\begin{equation}
    \mathcal H_{S_\alpha} = \bigoplus_{\mathcal H_j \in S_\alpha} \mathcal H_j,
\end{equation}
so that the overall space decomposes as
\begin{equation}
\mathcal H = \mathcal H_0 \oplus \left(\bigoplus_{\alpha} \mathcal H_{S_\alpha}\right).
\end{equation}

We then inspect each $S_\alpha$ separately. If $|S_\alpha|=1$, we simply test whether the unique irreducible component in $S_\alpha$ contains at least one stabilizer state, which can be verified by enumerating all $n$-qubit stabilizer states. If $|S_\alpha|>1$, we must check whether $\mathcal{H}_{S_\alpha}$ admits a decomposition into $|S_\alpha|$ irreducible subspaces, each containing a stabilizer state. To this end, we first identify all stabilizer states in $\mathcal{H}_{S_\alpha}$, denoted $\{\ket{s_1},\dots,\ket{s_m}\}$. For each $\ket{s_i}$, we compute its orbit under the group action:
\begin{equation}
    O_{s_i} = \{U\ket{s_i} : U \in G_\psi\}.
\end{equation}
Note that each orbit span $\mathrm{span}(O_{s_i})$ forms a $G_\psi$-invariant subspace. We then check whether there exist $|S_\alpha|$ orbits $O_{s_{i_1}},\dots,O_{s_{i_{|S_\alpha|}}}$ such that
\begin{equation}
    \mathcal H_{S_\alpha} = \mathrm{span}(O_{s_{i_1}}) \oplus \mathrm{span}(O_{s_{i_2}}) \oplus \cdots \oplus \mathrm{span}(O_{s_{i_{|S_\alpha|}}}),
\end{equation}
where we require each orbit span $\mathrm{span}(O_{s_{i_j}})$ is an irreducible representation of $G_\psi$. By repeating this procedure for all $S_\alpha$, we can verify the applicability of the single-copy scheme.
The complete procedure is summarized in Algorithm~\ref{alg:single_copy}.

\section{\label{sec:numerical_appen}Numerical simulation of Bell-measurement scheme}

\begin{figure}
    \centering
    \includegraphics[width=\linewidth]{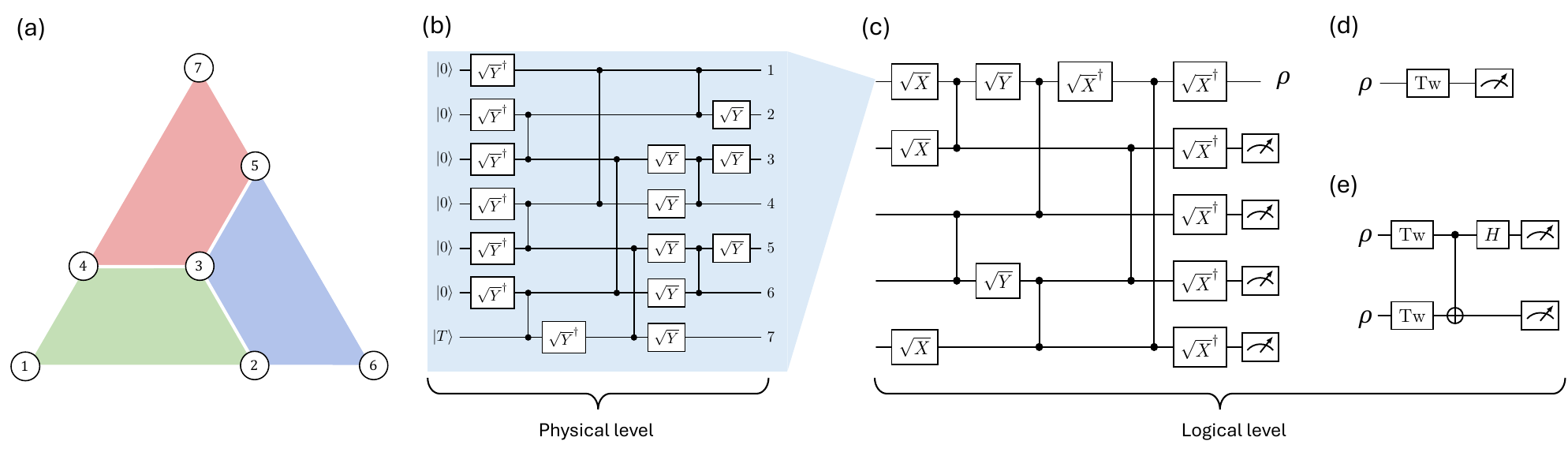}
    \caption{(a) Graphical description of the [[7, 1, 3]] color code. Each generator acts $X$ or $Z$ on red, blue, or green patches. (b) The non-fault-tolerant circuit on the physical level encodes a physical $\ket{T}$ state into [[7, 1, 3]] color code. (c) The 5-to-1 MSD circuit on the logical level. (d) The state tomography circuit on the logical level. (e) The Bell measurement circuit on the logical level.}
    \label{fig:Apdx_numerical_setup}
\end{figure}

In this section, we provide technical details for the numerical results presented in the main text. Here, we simulate the experimental setup of the Bell measurement scheme for logical $\ket{T}$ state. Our simulation setup closely follows the experimental method introduced in the recent experimental demonstration of logical magic state distillation~\cite{rodriguezExperimentalDemonstrationLogical2024}. With the distilled logical magic states encoded in the [[7, 1, 3]] color code, we run two different benchmarking circuits of state tomography and Bell measurement. Based on the simulation, we estimate the number of samples for the benchmarking task.

\subsection{[[7, 1, 3]] color code}

The [[7, 1, 3]] color code~\cite{Steane_1996} (also known as the Steane code) is a stabilizer code where each logical qubit consists of seven physical qubits. It has the following six stabilizer generators:
\begin{align}
    \begin{array}{cccccccccc}
        G_1 &= &I&I&I&X&X&X&X\\
        G_2 &= &I&X&X&I&I&X&X\\
        G_3 &= &X&I&X&I&X&I&X\\
        G_4 &= &I&I&I&Z&Z&Z&Z\\
        G_5 &= &I&Z&Z&I&I&Z&Z\\
        G_6 &= &Z&I&Z&I&Z&I&Z,
    \end{array}
    \label{eq:Steane_stabilizers}
\end{align}
Graphically, each generator acts $X$ or $Z$ on one of the color patches in Fig.~\ref{fig:Apdx_numerical_setup}(a). The corresponding logical states are given as
\begin{align}
    \ket{0_L}&=\frac{1}{\sqrt{8}}\left(\ket{0000000}+\ket{1010101}+\ket{0110011}+\ket{1100110}+\ket{0001111}+\ket{1011010}+\ket{0111100}+\ket{1101001}\right),\\
    \ket{1_L}&=\frac{1}{\sqrt{8}}\left(\ket{1111111}+\ket{0101010}+\ket{1001100}+\ket{0011001}+\ket{1110000}+\ket{0100101}+\ket{1000011}+\ket{0010110}\right).
\end{align}
A key advantage of this code is the transversal implementation of logical Clifford gates and logical Pauli measurements. Specifically, any Clifford gate $U$ can be realized logically by applying the same gate on each physical qubit:
\begin{equation}
    U_L = U^{\otimes 7},
\end{equation}
where $U_L$ denotes the implementation of $U$ at the logical level. Furthermore, since it implies $Z_L=Z^{\otimes 7}$, a logical computational-basis measurement can be performed simply by measuring each physical qubit. In the absence of physical errors, the logical measurement outcome is computed as the parity of the physical measurement outcomes $x \in \{0,1\}^7$:
\begin{equation}
    x_1 + \dots + x_7 \mod 2.
\end{equation}
In the presence of physical errors, however, additional processing—such as error detection or correction—is necessary, as will be described later.

\subsection{Circuit overview}

Our numerical simulation consists of three distinct stages: (1) encoding physical $\ket{T}$ states into the [[7, 1, 3]] color code [Fig.~\ref{fig:Apdx_numerical_setup}(b)], (2) logical-level implementation of the 5-to-1 MSD circuit [Fig.~\ref{fig:Apdx_numerical_setup}(c)], and (3) logical-level benchmarking circuits [Fig.~\ref{fig:Apdx_numerical_setup}(d)-(e)].

As illustrated in Fig.~\ref{fig:Apdx_numerical_setup}(b), the encoding circuit is non-fault-tolerant, resulting in significant initial noise in the logical magic state. The infidelity of this noisy state is reduced by the subsequent MSD process [Fig.\ref{fig:Apdx_numerical_setup}(c)], yielding the output state $\rho$. Here, the distilled state $\rho$ is accepted only if the four logical qubit measurement outcomes are 1, 0, 1, and 1, respectively (see Ref.~\cite{rodriguezExperimentalDemonstrationLogical2024} for details about the MSD circuit).

When $\rho$ is accepted, we perform benchmarking either by standard state tomography [Fig.\ref{fig:Apdx_numerical_setup}(d)] or via our proposed Bell measurement method [Fig.\ref{fig:Apdx_numerical_setup}(e)]. Each benchmarking circuit begins with a randomly chosen twirling gate from the set $\{I, U_0, U_0^\dagger\} \subset \mathcal{C}_1$, where
\begin{equation}
    U_0 = \frac{e^{i\pi/4}}{\sqrt{2}}
    \begin{pmatrix}
    1 & 1 \\
    i & -i
    \end{pmatrix},
\end{equation}
which cyclically transforms Pauli operators as $X \rightarrow Y \rightarrow Z \rightarrow X$,
\begin{align}
    U_0XU_0^\dagger = Y,\\
    U_0YU_0^\dagger = Z,\\
    U_0ZU_0^\dagger = X.
\end{align}
Consequently, the tomography circuit effectively measures $\rho$ in random Pauli bases $X$, $Y$, or $Z$. In our simulations, these logical-level Clifford circuits and measurements are realized via their transversal physical implementations.

To reflect realistic experimental imperfections, each physical-level operation is subjected to noise characterized by a single parameter $p$:
\begin{itemize}
    \item Single-qubit gates are followed by a depolarizing channel with average fidelity $1 - p/5$.
    \item Two-qubit gates are followed by a depolarizing channel with average fidelity $1 - p$.
    \item Qubit initialization includes a depolarizing channel with average fidelity $1 - p/2$.
    \item Measurements are simulated with classical bit-flip errors occurring with probability $p/2$.
\end{itemize}
The single-qubit and two-qubit depolarizing channels with average fidelity $1-q$ are defined as:
\begin{align}
    \mathcal{N}_{\rm depo, 1}(\rho)
    &= (1-2q)\rho + qI\\
    &= \left(1-\frac{3}{2}q\right)\rho+\frac{q}{2}\sum_{P\in\{X,Y,Z\}}P\rho P,\\
    \mathcal{N}_{\rm depo, 2}(\rho)
    &= \left(1-\frac{4}{3}q\right)\rho + \frac{q}{3}I^{\otimes 2}\\
    &= \left(1-\frac{5}{4}q\right)\rho+\frac{q}{12}\sum_{P\in\{I,X,Y,Z\}^{\otimes2}\setminus\{II\}}P\rho P.
\end{align}

\subsection{Error detection}

Combining the previously described components, the resulting simulation circuits comprise $7 \times 5 = 35$ physical qubits for state tomography and $7 \times 5 \times 2 = 70$ physical qubits for Bell measurement, with all qubits measured at the end. Based on these measurement outcomes $x \in \{0,1\}^{35}$ (state tomography) or $x \in \{0,1\}^{70}$ (Bell measurement), we perform an error detection. Specifically, we check the consistency of the outcomes with the stabilizer generators defined in Eq.~\eqref{eq:Steane_stabilizers}. Since phase errors ($Z$-errors) do not influence the measurement outcomes, the stabilizer generators $G_1, G_2, G_3$ (comprising only $X$ operators) are irrelevant for error detection. Thus, error detection is performed exclusively with the stabilizers $G_4, G_5, G_6$, each containing only $Z$ operators. We accept the measurement outcomes $x$ if none of these stabilizer generators detect an error; otherwise, we discard them. All reported logical infidelities $\epsilon$ and benchmarking results are calculated after this error-detection step.

The MSD and benchmarking circuits exclusively utilize transversal gates and measurements, ensuring fault tolerance. Given that the [[7, 1, 3]] color code can detect up to two errors, the logical error rates after error detection are suppressed to the third order in $p$. Meanwhile, the 5-to-1 MSD circuit reduces the infidelity of the distilled magic state to second order in $p$. Consequently, logical errors introduced during benchmarking are negligible for small physical error rates $p$.

\subsection{Calculating sampling overhead}

\begin{figure}[t]
    \centering
    \includegraphics[width=1\linewidth]{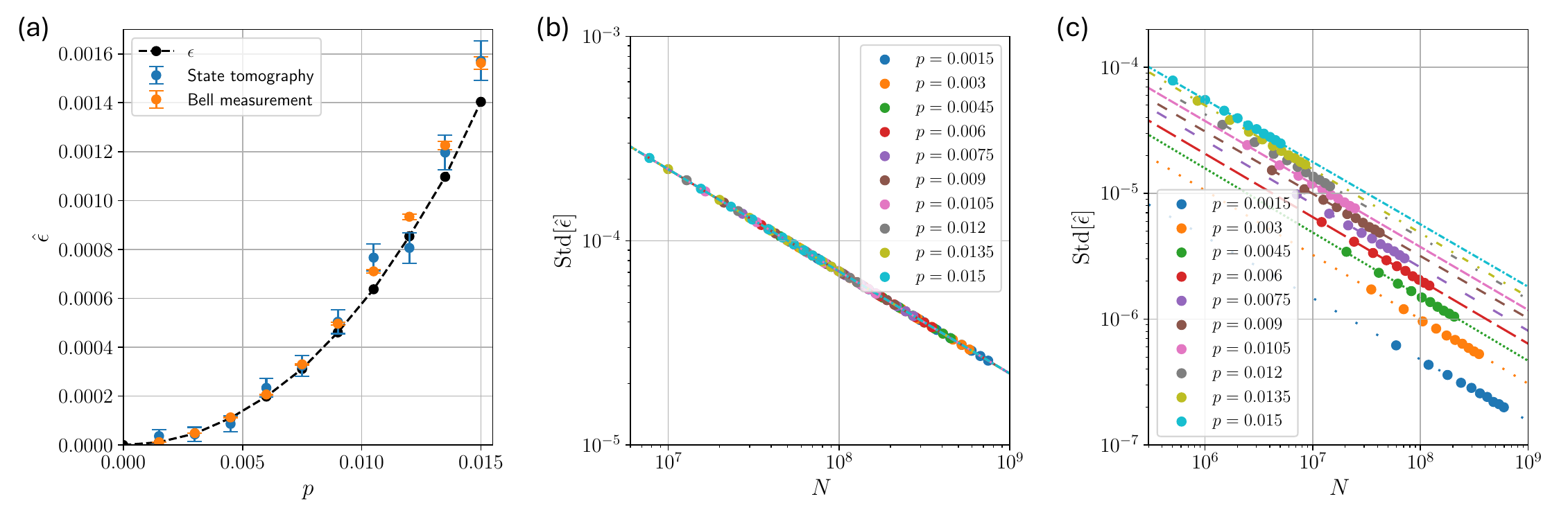}
    \caption{(a) Infidelity $\epsilon$ of logical state $\rho$ compared with estimated infidelity $\hat{\epsilon}$ obtained from standard state tomography and the Bell measurement benchmarking scheme. The black dashed line shows a power-law fit of $\epsilon$ as $y=8.11x^{2.07}$. (b), (c) Standard deviation of estimated infidelity $\hat{\epsilon}$ as a function of the number of consumed copies of $\rho$ for state tomography (b) and Bell measurement (c). Dashed and dotted lines indicate power-law fits for respective error rates $p$.}
    \label{fig:Apdx_numerical}
\end{figure}

Fig.~\ref{fig:Apdx_numerical}(a) presents numerical simulation results comparing the true infidelity $\epsilon$ of the logical state $\rho$ with the estimated infidelity $\hat{\epsilon}$ obtained via standard state tomography and the Bell measurement benchmarking scheme. The full simulation circuits were repeated $5.76\times 10^8$ to $5.76\times 10^9$ times for state tomography method and $1.8234\times 10^9$ to $1.8234\times 10^{10}$ times for Bell measurement method to obtain the estimations~\footnote{The extra overhead for Bell measurement simulation comes from the requirement for the simultaneous success of two distilled magic state in our simulation, which in practice is unnecessary if we have an ability to store distilled magic states.}.

To analyze the asymptotic behavior of sampling overhead, we consider the number $N$ of copies of $\rho$ consumed and investigate how the standard deviation of the estimated infidelity decreases with increasing $N$. Figs.~\ref{fig:Apdx_numerical}(b) and (c) illustrate these relationships for the standard state tomography and Bell measurement, respectively. To quantify the sampling overhead $N$ required for benchmarking with a target multiplicative precision $r$, we assume the estimator $\hat{\epsilon}$ to approach a random variable from a normal distribution due to the central limit theorem, such that ${\rm Std}[\hat{\epsilon}] = r\epsilon$ corresponds to achieving the benchmarking task with approximately $68\%$ confidence.

As discussed in the main text, the standard state tomography samples from a Bernoulli distribution ${\rm Bern}(c + O(\epsilon))$ for some constant $0 < c < 1$. Consequently, the standard deviation of the estimator is
\begin{equation}
{\rm Std}[\hat{\epsilon}] = \sqrt{(c + O(\epsilon))(1 - c - O(\epsilon))/N} \approx \sqrt{c(1 - c)/N},
\end{equation}
which is approximately independent of $\epsilon$. This behavior is observed in Fig.~\ref{fig:Apdx_numerical}(b), where the standard deviations appear nearly identical across different values of $p$ (and hence different $\epsilon$).
In contrast, the Bell measurement scheme samples from ${\rm Bern}(\Theta(\epsilon))$, leading to
\begin{equation}
{\rm Std}[\hat{\epsilon}] = \sqrt{\Theta(\epsilon)/N},
\end{equation}
which decreases with $\epsilon$. This $\epsilon$-dependent behavior is clearly shown in Fig.~\ref{fig:Apdx_numerical}(c), where smaller values of $p$ result in smaller standard deviations.

To systematically extract the sampling overhead $N$ required to achieve a target precision, we perform a power-law fit for each value of $p$ (and hence fixed $\epsilon$):
\begin{equation}
    {\rm Std}[\hat{\epsilon}] = a \times N^{b},
\end{equation}
where $a$ and $b$ are fitting parameters obtained from the numerical data. Plugging this relation in ${\rm Std}[\hat{\epsilon}]=r\epsilon$ yields the estimated sample size needed to achieve multiplicative precision $r$,
\begin{equation}
    N = \left(\frac{r\epsilon}{a}\right)^{1/b}.
\end{equation}
The resulting estimates of sampling overhead are reported in the main text.

\section{\label{sec:numerical_appen_additional}Numerical simulation of single-copy scheme}

In addition to the Bell measurement scheme, we provide another numerical result for the single-copy scheme, which is designed to benchmark multi-qubit magic states. Here, we simulate the experimental setup for benchmarking the logical fidelity of the $\ket{\rm CCZ}$ state. This simulation setup follows the experimental method realized with trapped ion devices in Ref.~\cite{Wang2024_832colorcode}. Contrary to the previous simulation for $\ket{T}$ state, we do not perform magic state distillation. Instead, we utilize error-detection circuits to encode $\ket{\rm CCZ}$ state into the [[8, 3, 2]] color code with high fidelity. After obtaining high-fidelity logical $\ket{\rm CCZ}$ state, we run the benchmarking circuits and analyze the results. Notably, Ref.~\cite{Wang2024_832colorcode} demonstrated the fault-tolerant preparation of logical $\ket{\rm CCZ}$ state in the [[8, 3, 2]] color code and implemented fault-tolerant logical $ZZX$ measurements. In combination with the logical twirling gates for the ${\rm CCZ}$ state described in Appendix~\ref{sec:twirling}---which can be implemented with transversal $X$, ${\rm CNOT}$, and ${\rm CCZ}$ gates---these operations enable the application of our single-copy benchmarking protocol.

\subsection{[[8, 3, 2]] color code}

The [[8, 3, 2]] color code is a stabilizer code encoding three logical qubits using eight physical qubits, detecting any single-qubit error~\cite{Campbell2016smallestcolourcode}. It has the following five stabilizer generators:
\begin{align}
    \begin{array}{ccccccccccc}
        G_1 &= &Z&Z&Z&Z&I&I&I&I\\
        G_2 &= &Z&Z&I&I&Z&Z&I&I\\
        G_3 &= &Z&I&Z&I&Z&I&Z&I\\
        G_4 &= &Z&Z&Z&Z&Z&Z&Z&Z\\
        G_5 &= &X&X&X&X&X&X&X&X\\
    \end{array}
    \label{eq:832_stabilizers}
\end{align}
Graphically, $G_1$, $G_2$, and $G_3$ acts $Z$ on the three faces (red, green, and blue, respectively) of the cube in Fig.~\ref{fig:Apdx_numerical_setup2}(a), and $G_4$ and $G_5$ act $Z$ and $X$ on all qubits, respectively. The corresponding logical $X$ and $Z$ operators are given as
\begin{align}
    \begin{array}{ccccccccccc}
        (X_1)_L &= &X&X&X&X&I&I&I&I\\
        (X_2)_L &= &X&X&I&I&X&X&I&I\\
        (X_3)_L &= &X&I&X&I&X&I&X&I\\
        (Z_1)_L &= &Z&I&I&I&Z&I&I&I\\
        (Z_2)_L &= &Z&I&Z&I&I&I&I&I\\
        (Z_3)_L &= &Z&Z&I&I&I&I&I&I
    \end{array}
\end{align}

Several logical gates can be implemented transversally in this code. As noticed above, any Pauli operator can be implemented transversally. In particular, it also admits a transversal ${\rm CCZ}$ gate within a code block, given by
\begin{equation}
    {\rm CCZ}_L = T \otimes T^\dagger \otimes T^\dagger \otimes T \otimes T^\dagger \otimes T \otimes T \otimes T^\dagger.
\end{equation}
In addition, a logical ${\rm CNOT}$ gate within a code block can be written as SWAP operators:
\begin{align}
    ({\rm CNOT}_{1,2})_L &= {\rm SWAP}_{1,5} {\rm SWAP}_{2,6},\\
    ({\rm CNOT}_{2,1})_L &= {\rm SWAP}_{1,3} {\rm SWAP}_{2,4},\\
    ({\rm CNOT}_{1,3})_L &= {\rm SWAP}_{1,5} {\rm SWAP}_{3,7},\\
    ({\rm CNOT}_{3,1})_L &= {\rm SWAP}_{1,2} {\rm SWAP}_{3,4},\\
    ({\rm CNOT}_{2,3})_L &= {\rm SWAP}_{1,3} {\rm SWAP}_{5,7},\\
    ({\rm CNOT}_{3,2})_L &= {\rm SWAP}_{1,2} {\rm SWAP}_{5,6}.
\end{align}
Rather than physically swapping the corresponding qubits, in practice, one can implement these logical gates by relabeling physical qubits. This approach does not require physical operation or gate application and thus avoids introducing additional errors.

\subsection{Circuit overview}

\begin{figure}
    \centering
    \includegraphics[width=\linewidth]{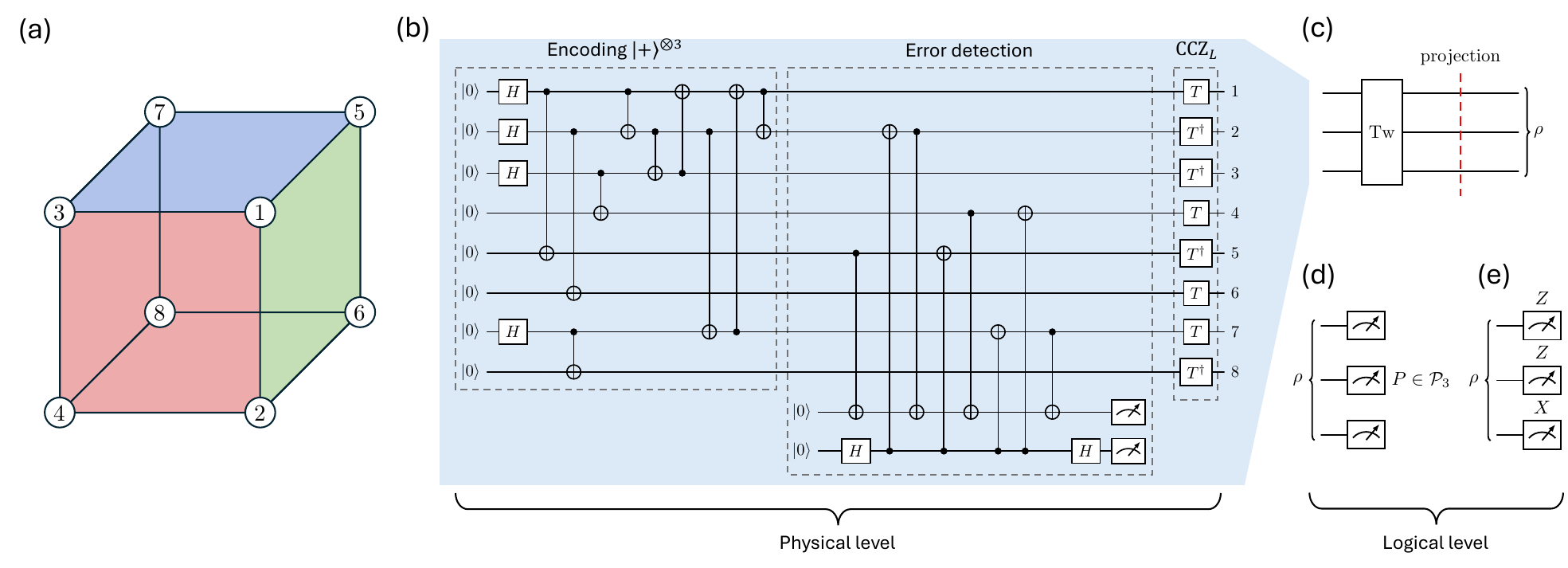}
    \caption{(a) Graphical description of the [[8, 3, 2]] color code. Each generator acts $Z$ on red, green, blue faces of the cube, or $Z$/$X$ on all qubits. (b) The physical-level circuit proposed in Ref.~\cite{Wang2024_832colorcode} that prepares the logical $\ket{{\rm CCZ}_L}$ state in [[8,3,2]] color code. The first dashed box non-fault-tolerantly encodes $\ket{+}^{\otimes 3}$ state, and the second dashed box performs error-detection. Finally, the last dashed box performs logical ${\rm CCZ}$ gate transversally to obtain $\ket{{\rm CCZ}_L}$. (c) The state twirling and projection circuit on the logical level. (d) The state tomography circuit on the logical level. (e) The single-copy benchmarking circuit on the logical level.}
    \label{fig:Apdx_numerical_setup2}
\end{figure}

Our numerical simulation consists of three distinct stages: (1) encoding $\ket{\rm CCZ}$ state into the [[8, 3, 2]] color code [Fig.~\ref{fig:Apdx_numerical_setup2}(b)], (2) state twirling and projection to the code space [Fig.~\ref{fig:Apdx_numerical_setup2}(c)], and (3) logical-level benchmarking circuits [Fig.~\ref{fig:Apdx_numerical_setup2}(d)-(e)]. As in Appendix~\ref{sec:numerical_appen}, each physical-level operation is subjected to noise characterized by a single parameter $p$ with the same noise model.

First, we obtain high-fidelity $\ket{\rm CCZ}$ encoded in the [[8, 3, 2]] color code using the circuit shown in Fig.~\ref{fig:Apdx_numerical_setup2}(b). To this end, we first generate a noisy logical $\ket{+}^{\otimes 3}$ state by applying a non-fault-tolerant encoding circuit. Here, rather than performing a magic state distillation as in previous section, we perform an error-detection circuit to detect errors, closely following Ref.~\cite{Wang2024_832colorcode}. Specifically, by introducing two ancillary qubits, we measure
\begin{align}
    \begin{array}{ccccccccccc}
        (X_1)_L(X_3)_L &= &I&X&I&X&X&I&X&I,\\
        G_1 G_3 &= &I&Z&I&Z&Z&I&Z&I,
    \end{array}
\end{align}
each of which stabilizes the logical $\ket{+}^{\otimes 3}$ state. If both measurements are 0s, we accept the resulting state, and discard otherwise. Although this error-detection step does not measure all stabilizer generators, it can detect all high-weight propagated errors~\cite{Wang2024_832colorcode, Reichardt2020faulttolerantquantum}. Afterwards, we apply a transversal logical ${\rm CCZ}$ gate to obtain the logical $\ket{\rm CCZ}$ state (note that $\ket{\rm CCZ} = {\rm CCZ}\ket{+}^{\otimes 3}$).

With the encoded logical $\ket{\rm CCZ}$ state, we perform state twirling on the logical level~[Fig.~\ref{fig:Apdx_numerical_setup2}(c)]. As described in Appendix~\ref{sec:twirling}, the twirling group for $\ket{\rm CCZ}$ state is given by
\begin{equation}
    G_{\rm CCZ} = \{{\rm CCZ} \cdot U \cdot {\rm CCZ}: U\in \langle X_i, {\rm CNOT}_{j,k} \rangle \text{ for }i,j,k = 1,2,3\},
\end{equation}
and thus generated by $X$ gates, ${\rm CNOT}$ gates, and ${\rm CCZ}$ gates. Since all these gates can be implemented transversally in the [[8, 3, 2]] color code, the state twirling is realized by applying these logical gates transversally. With these logical gates, the state twirling is realized by randomly choosing a twirling gate from $G_{\rm CCZ}$. For ease of analysis, we also perform a perfect projection (equivalent to a full round of noiseless error-detection) to the code space so that the resulting state $\rho$ lies within the code space. This step ensures that there is no $O(p)$ population that is left out of the logical subspace as it could affect the following measurement outcomes. After this projection, the infidelity of the resultant state should scale as $O(p^2)$. 

After obtaining high-fidelity logical $\ket{\rm CCZ}$ state, we perform benchmarking either by standard state tomography [Fig.~\ref{fig:Apdx_numerical_setup2}(d)] or via our proposed single-copy scheme [Fig.~\ref{fig:Apdx_numerical_setup2}(e)]. Here, the infidelity of the input state $\rho$ is second order in $p$ due to the error-detection step. Meanwhile, if the gates and measurements in the benchmarking circuit are noisy, the errors introduced in the benchmarking circuit is also second order in $p$, and thus benchmarking would fail. To avoid this issue, we assume that the benchmarking circuits are ideal without any errors~\footnote{Note that we did not have such a issue in Appendix~\ref{sec:numerical_appen} as the error in the benchmarking circuit was negligible compared to the infidelity}. With this setup, we measure $\rho$ with all Pauli operators on three logical qubits (except $I^{\otimes 3}$) using the same number of sampling circuits for the standard state tomography. In the single-copy scheme, we perform the logical $ZZX$ measurement.

Here, we remark that this assumption of ideal benchmarking circuits is practically relevant to several experimental platforms. For instance, in trapped-ion systems, measurement error can be greatly suppressed by performing repetitive quantum non-demolition (QND) measurements~\cite{Hume2007highfidelity,Myerson2008highfidelity}, rendering it negligible compared with gate errors.

\subsection{Calculating sampling overhead and results}

\begin{figure}
    \centering
    \includegraphics[width=1\linewidth]{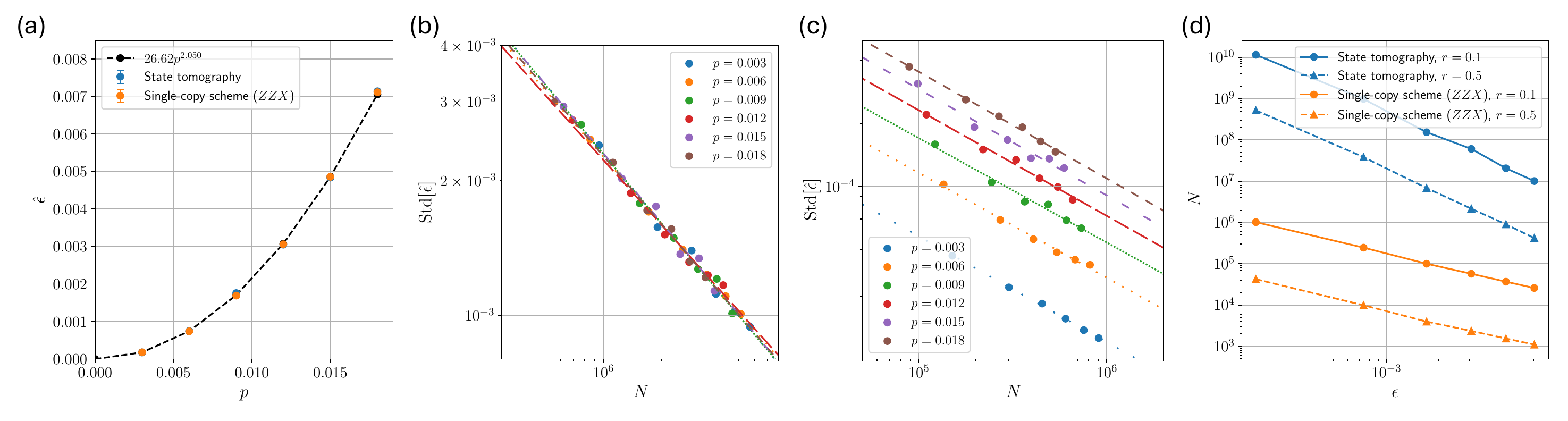}
    \caption{(a) Infidelity $\epsilon$ of logical state $\rho$ compared with estimated infidelity $\hat{\epsilon}$ obtained from standard state tomography and the single-copy benchmarking scheme. The black dashed line shows a power-law fit of $\epsilon$ as $y=26.62x^{2.05}$. (b), (c) Standard deviation of estimated infidelity $\hat{\epsilon}$ as a function of the number of consumed copies of $\rho$ for state tomography (b) and single-copy scheme with $ZZX$ measurement (c). Dashed and dotted lines indicate power-law fits for respective error rates $p$. (d) Sampling overhead $N$ for the standard state tomography and single-copy scheme with $CCZ$ measurement to achieve multiplicative precision $r=0.1$ and $r=0.5$ with $68\%$ confidence.}
    \label{fig:Apdx_numerical2}
\end{figure}

For calculating the sampling overhead, we follow the same procedure as in Appendix~\ref{sec:numerical_appen}. First, Fig.~\ref{fig:Apdx_numerical2}(a) presents numerical simulation results of the estimated infidelity $\hat{\epsilon}$ obtained via standard state tomography and the single-copy benchmarking scheme. The full simulation circuits were repeated $3{,}048{,}192{,}000$ times for state tomography method and $483{,}840{,}000$ times for single-copy method with $ZZX$ measurement to obtain the estimations. Due to the absence of errors in the benchmarking circuits and sufficient number of samples, both methods yield accurate estimations of the true infidelity $\epsilon$ of the logical state $\rho$. We fit a power-law curve to the data, resulting in $\epsilon \approx 26.62 p^{2.050}$.

To analyze the asymptotic behavior of sampling overhead, we again consider the number $N$ of copies of $\rho$ consumed and investigate how the standard deviation of the estimated infidelity decreases with increasing $N$. Figs.~\ref{fig:Apdx_numerical2}(b) and (c) illustrate these relationships for the standard state tomography and single-copy scheme with $ZZX$ measurement, respectively. By performing power-law fits for each value of $p$, we extract the sampling overhead $N$ required to achieve a target multiplicative precision $r$.

Finally, Fig.~\ref{fig:Apdx_numerical2}(d) compares the estimated sampling overheads for both methods. The single-copy scheme demonstrates a significant reduction in sampling overhead compared to standard state tomography, particularly at lower infidelities $\epsilon$. This result validates the advantage of the single-copy benchmarking scheme for multi-qubit magic states.

\section{\label{sec:ancilla}Ancillary qubits provide no advantage}

In this section, we show that additional ancillary qubits do not change the sample complexity of the benchmarking schemes, under the restriction that only Clifford gates and Pauli measurements are available. Specifically, we consider a circuit where an $n$-qubit state $\rho$ undergoes a Clifford gate with an additional $m$ ancillary qubits, all initialized in the state $\ket{0^m}$. That is, we apply a Clifford gate $U \in \mathcal{C}_{n+m}$ to the state $\rho \otimes \ketbra{0^m}{0^m}$ and then measure all $n+m$ qubits in the computational basis.

We show that this circuit can be \textit{simulated} by an $n$-qubit circuit that does not use any ancillary qubits. Specifically, we construct an $n$-qubit Clifford gate $V \in \mathcal{C}_n$ such that measuring $V \rho V^\dagger$ in the computational basis, followed by some classical post-processing, yields the same output distribution as the original circuit. We begin with the simplest case where the original circuit includes only one ancillary qubit $m=1$.

\begin{theorem}
    Consider an $(n+1)$-qubit Clifford circuit applying $U \in \mathcal{C}_{n+1}$ jointly on an $n$-qubit state $\rho$ and a single ancilla initialized to $\ket{0}$, followed by computational basis measurements. Then, there exists an $n$-qubit Clifford circuit applying $V \in \mathcal{C}_n$ solely on $\rho$, followed by computational basis measurements and classical post-processing, that yields an identical measurement distribution.
    \label{thm:single-ancilla}
\end{theorem}

\begin{figure}
    \centering
    \includegraphics[width=\columnwidth]{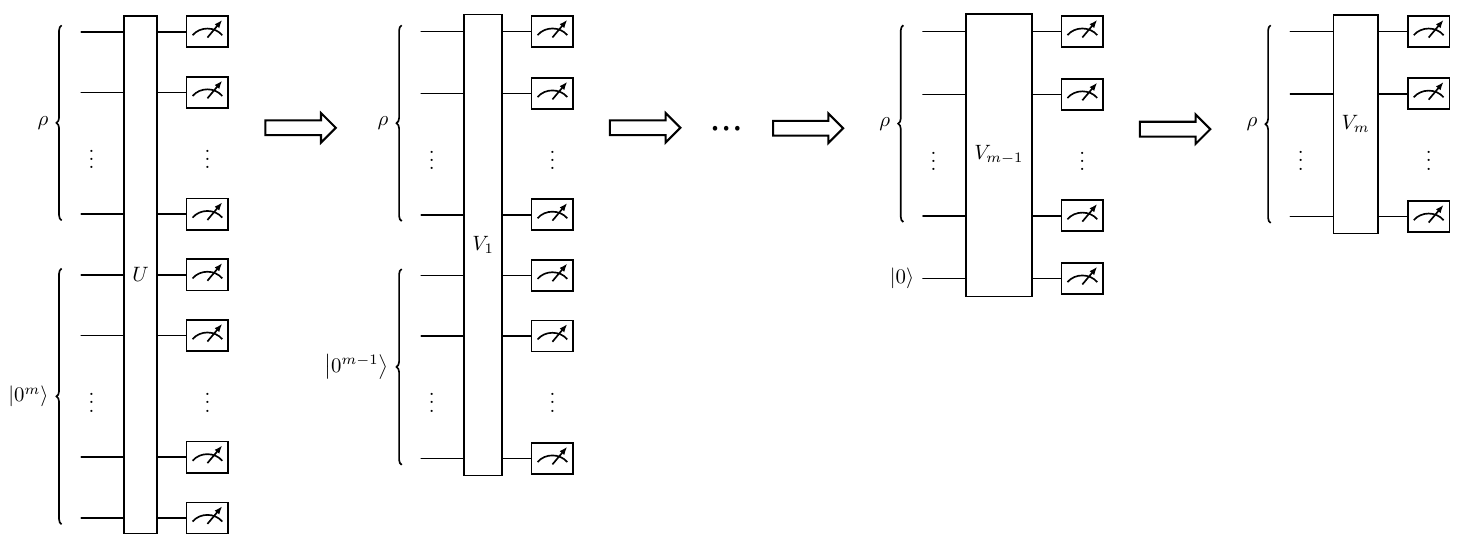}
    \caption{Recursive applications of Theorem~\ref{thm:single-ancilla} leads to Corollary~\ref{cor:multi-ancilla}.}
    \label{fig:removing_ancillas}
\end{figure}

By recursively applying Theorem~\ref{thm:single-ancilla}, it follows immediately that a circuit with any number of ancillary qubits can be simulated by a circuit without ancillas. As Fig.~\ref{fig:removing_ancillas} describes, we can eliminate one ancillary qubit at a time, replacing the circuit with a Clifford circuit without the last ancilla and classical post-processing, while preserving the output probability distribution. Repeating this procedure $m$ times, we ultimately obtain a circuit acting only on the $n$-qubit state $\rho$.

\begin{corollary}
    Consider an $(n+m)$-qubit Clifford circuit applying $U \in \mathcal{C}_{n+1}$ jointly on an $n$-qubit state $\rho$ and a single ancilla initialized to $\ket{0^m}$, followed by computational basis measurements. Then, there exists an $n$-qubit Clifford circuit applying $V \in \mathcal{C}_n$ solely on $\rho$, followed by computational basis measurements and classical post-processing, that yields an identical measurement distribution.
    \label{cor:multi-ancilla}
\end{corollary}

Therefore, the presence of additional ancillas cannot reduce the sample complexity. We conclude this section by proving Theorem~\ref{thm:single-ancilla}.

\begin{proof}[Proof of Theorem~\ref{thm:single-ancilla}]
    To begin with, we introduce the notations that we use throughout this proof. We denote the $n$ qubit Pauli group as $\mathcal{P}_n$. In addition, let $I_S$, $I_A$, and $I_{SA}$ denote the identity acting on the first $n$ qubits for the main system, the last ancilla qubit, and all $n+1$ qubits, respectively. Finally, for $j=1,\dots,n+1$ $X_j$ denotes the Pauli operator acts $X$ on the $j$-th qubit and acts $I$ on the other qubits. We denote $Y_j$ and $Z_j$ similarly.
    
    For the original circuit with the ancilla, let $P(x)$ denote the probability that the measurement outcome is $x\in \{0,1\}^{n+1}$. Then,
    \begin{align}
       P(x) &= \langle x|U(\rho\otimes \ketbra{0}{0}) U^\dagger |x\rangle\\
       &= \Tr\left[\rho (I_S \otimes \bra{0})U^\dagger\ketbra{x}{x}U(I_S\otimes \ket{0})\right],
    \end{align}
    where $I_S$ is the identity operator on the first $n$ qubits. Thus, it is useful to analyze $U^\dagger\ketbra{x}{x}U$ first. Note that for all $x\in \{0,1\}^{n+1}$, $\ketbra{x}{x}=\prod_{i=1}^{n+1}\frac{I_{SA}+(-1)^{x_i}Z_i}{2}$, where $x_i$ denotes the $i$-th element of $x$. Then, we have
    \begin{align}
        U^{\dagger}\ketbra{x}{x}U=\prod_{i=1}^{n+1}\frac{I_{SA}+(-1)^{x_i}g_i}{2}
    \end{align}
    where $g_i=U^{\dagger}Z_iU\in\mathcal{P}_{n+1}$ for all $i=1,\dots,n+1$. Since $Z_1, \dots, Z_{n+1}$ are mutually commuting and independent, so are $g_1, \dots, g_{n+1}$. Therefore, we can define a stabilizer group $S$, an Abelian subgroup of $\mathcal{P}_n$ generated by $(-1)^{x_i}g_i$'s, i.e.,
    \begin{equation}
        S = \langle (-1)^{x_1}g_1, \dots, (-1)^{x_{n+1}}g_{n+1} \rangle.
    \end{equation}
    Then, as $g_i^2=I$ for all $i=1,\dots,n+1$, we have
    \begin{equation}
        U^{\dagger}\ketbra{x}{x}U=\frac{1}{2^{n+1}}\sum_{P\in S}P.
    \end{equation}
    It shows that the state only depends on the group $S$, independent of the choice of the generator. Therefore, if $\{(-1)^{y_1}g_1',
    \dots,(-1)^{y_{n+1}}g_{n+1}'\}$ also generates $S$ for some $g_1',\dots,g_{n+1}'\in \mathcal{P}_{n+1}$ and $y \in \{0,1\}^{n+1}$, then we can also denote $U^{\dagger}\ketbra{x}{x}U=\prod_{i=1}^{n+1}\frac{I_{SA}+(-1)^{y_i}g_i}{2}$.

    We categorize the proof into four cases based on the action of $g_1,\dots,g_{n+1}$ on the ancilla (the $(n+1)$-th qubit): (1) All $g_i$ act as $I$ or $Z$ on the ancilla. (2) All $g_i$ act as $I$ or $X$ on the ancilla. (3) All $g_i$ act as $I$ or $Y$ on the ancilla. (4) There exist at least two $g_i$'s acting non-trivially on the ancilla with different Pauli operators. We now show, for each case 1-4, the existence of a corresponding $V\in\mathcal{C}_n$ and classical post-processing that yield measurement distributions identical to the original circuit.
    
    \textit{Case 1}: All $g_i$ act as $I$ or $Z$ on the ancilla. In this case, there exists $g \in \{g_1,\dots,g_{n+1}\}$ that acts $Z$ on the ancilla, since otherwise, the rank of $\prod_{i=1}^{n+1}\frac{I_{SA}+(-1)^{x_i}g_i}{2}$ is more than two, while the rank of $U^\dagger \ketbra{x}{x}U$ should always be one. Without loss of generality, let $g_{n+1}$ act $Z$ on the ancilla.

    Given $\{(-1)^{x_1}g_1,\dots,(-1)^{x_{n+1}}g_{n+1}\}$ as the generating set of $S$, we have another generating set by multiplying $(-1)^{x_{n+1}}g_{n+1}$ to the other generators that act also $Z$ on the ancilla, so that $(-1)^{x_{n+1}}g_{n+1}$ is the only generator that acts nontrivially on the ancilla. Specifically, for some one-to-one and onto function $\alpha:\{0,1\}^{n+1}\rightarrow\{0,1\}^{n+1}$ and some $g'_1,\dots,g'_{n+1} \in \mathcal{P}_{n+1}$, we have the following generators of $S$:
    \begin{equation}
        S = \langle (-1)^{\alpha(x)_1}g'_1, \dots, (-1)^{\alpha(x)_{n}}g'_{n}, (-1)^{\alpha(x)_{n+1}}g_{n+1} \rangle,
    \end{equation}
    where $g'_1,\dots,g'_n$ acts trivially on the ancilla for $j=1,\dots,n$. Therefore, we can denote $g_{n+1}=h_{n+1}\otimes Z$ for some $h_{n+1} \in \mathcal{P}_n$, and $g'_j=h_j\otimes I$ for $j=1,\dots ,n$ where $h_1,\dots,h_n \in \mathcal{P}_{n}$ are $n$ independent and mutually commuting Pauli operators in $\mathcal{P}_n$. We can further simplify the generators by using the fact that $\{h_1,\dots,h_{n+1}\}$ is not an independent set. This is because the size of the generating set of an Abelian subgroup of $\mathcal{P}_n$ cannot be larger than $n$ (see, e.g., Ref.~\cite{Nielsen_Chuang_2010}). Therefore, $h_{n+1}$ can be written as a product of some of $h_1,\dots, h_n$, so multiplying some of $(-1)^{\alpha(x)_1}g'_1 \dots (-1)^{\alpha(x)_n}g'_n$ to $(-1)^{\alpha(x)_{n+1}}g_{n+1}$ so that the last generator acts nontrivially only on the ancilla. Therefore, with the one-to-one and onto function $\beta: \{0,1\}^{n+1} \rightarrow \{0,1\}^{n+1}$ corresponding to this multiplication and $\gamma = \beta \circ \alpha$, we have
    \begin{equation}
        S = \langle (-1)^{\gamma(x)_1}g'_1, \dots, (-1)^{\gamma(x)_{n}}g'_{n}, (-1)^{\gamma(x)_{n+1}}g'_{n+1} \rangle,
    \end{equation}
    where $g'_{n+1} = Z_{n+1}$.

    Now, let $V \in \mathcal{C}_n$ be a $n$-qubit Clifford gate such that $V^{\dagger} Z_i V = h_i$ for all $i=1,\dots,n$. Then,
    \begin{align}
        U^{\dagger} \ketbra{x}{x} U
        &= \prod_{i=1}^{n+1}\frac{I_{SA}+(-1)^{\gamma(x)_i}g'_i}{2}\\
        &= \left(\prod_{j=1}^{n}\frac{I_S+(-1)^{\gamma(x)_j}h_j}{2}\otimes I_A\right)\frac{I_{SA}+(-1)^{\gamma(x)_{n+1}}Z_{n+1}}{2}\\
        &= \left(\prod_{j=1}^{n}\frac{I_{S}+(-1)^{\gamma(x)_j}V^\dagger Z_j V}{2}\right)\otimes\frac{I_A+(-1)^{\gamma(x)_{n+1}}Z}{2}\\
        &= \left(V^\dagger\ketbra{\gamma(x)_{1:n}}{\gamma(x)_{1:n}} V \right) \otimes \ketbra{\gamma(x)_{n+1}}{\gamma(x)_{n+1}}
    \end{align}
    With this notation, the output distribution $P$ is
    \begin{align}
        P(x) &=
        \begin{cases}
            \langle\gamma(x)_{1:n}|V\rho V^{\dagger}|\gamma(x)_{1:n}\rangle,\quad&\text{if $\gamma(x)_{n+1}=0$}\\
            0,\quad&\text{if $\gamma(x)_{n+1}=1$}
        \end{cases}
    \end{align}

    Now, it is straightforward to obtain the same output distribution using a Clifford circuit without the ancillary qubit and post-processing. Consider a circuit with the input of $\rho$ without an ancilla, and we apply $V$ and then measure all $n$ qubits in the computational basis. Let us denote the output distribution of this new circuit as $Q$. Then,
    \begin{align}
        Q(y) = \langle y|V\rho V^{\dagger}|y\rangle.
    \end{align}
    for all $y\in\{0,1\}^n$. Once we get an outcome $y$, we process $y \mapsto x = \gamma^{-1}(y0)$. Let the distribution of $x$ after the post-processing be $Q'$. Then,
    \begin{align}
        Q'(x)
        &=
        \begin{cases}
            Q(\gamma(x)_{1:n}),\quad&\text{if }\gamma(x)_{n+1}=0\\
            0,\quad&\text{if }\gamma(x)_{n+1}=1
        \end{cases}
        \\
        &=
        \begin{cases}
            \langle\gamma(x)_{1:n}|V\rho V^{\dagger}|\gamma(x)_{1:n}\rangle,\quad&\text{if }\gamma(x)_{n+1}=0\\
            0,\quad&\text{if }\gamma(x)_{n+1}=1
        \end{cases}
    \end{align}
    which is identical to the distribution of the original circuit $P(x)$. 
    
    \textit{Case 2}: All $g_i$ act as $I$ or $X$ on the ancilla. By the same reasoning of the \textit{Case 1}, we can choose the generators for the stabilizer group as
    \begin{equation}
        S = \langle (-1)^{\gamma(x)_1}g'_1, \dots, (-1)^{\gamma(x)_{n+1}}g'_{n+1} \rangle,
    \end{equation}
    for some one-to-one and onto function $\gamma:\{0,1\}^{n+1} \rightarrow \{0,1\}^{n+1}$ and $g'_1, \dots, g'_{n+1}$ in the forms of $g_j = h_j \otimes I$ with $h_j \in \mathcal{P}_n$ for $j=1,\dots,n$, and $g'_{n+1} = X_{n+1}$. Let $V \in \mathcal{C}_n$ be a $n$-qubit Clifford gate such that $V^{\dagger} Z_i V = h_i$ for all $i=1,\dots,n$. Then, similarly as before,
    \begin{align}
        U^{\dagger} \ketbra{x}{x} U
        &= \left(\prod_{j=1}^{n}\frac{I_{S}+(-1)^{\gamma(x)_j}V^\dagger Z_j V}{2}\right)\otimes\frac{I_A+(-1)^{\gamma(x)_{n+1}}X}{2}\\
        &= \left(V^\dagger\ketbra{\gamma(x)_{1:n}}{\gamma(x)_{1:n}} V \right) \otimes\frac{I_A+(-1)^{\gamma(x)_{n+1}}X}{2}.
    \end{align}
    Since $\langle0|X|0\rangle=0$, the output distribution $P$ is
    \begin{equation}
        P(x) =\frac{1}{2}\langle\gamma(x)_{1:n}|V\rho V^{\dagger}|\gamma(x)_{1:n}\rangle,
    \end{equation}
    for all $x\in \{0,1\}^{n+1}$

    Again, the circuit with the input of $\rho$ that applies $V$ and measuring every qubit can produce the same output distribution with post-processing. Let the output distribution of this new circuit be $Q$. Then,
    \begin{equation}
        Q(y) = \langle{y}|V\rho V^{\dagger}|{y}\rangle,
    \end{equation}
    for all $y\in\{0,1\}^n$. Once we get an outcome $y$, we process it as follows:
    \begin{equation}
        y \mapsto
        x = 
        \begin{cases}
            \gamma^{-1}(y0), \quad &\text{with probability 1/2}\\
            \gamma^{-1}(y1), &\text{with probability 1/2}
        \end{cases}
    \end{equation}
    Let the distribution of $x$ after this post-processing be $Q'$. Then,
    \begin{align}
        Q'(x) &= \frac{1}{2}Q(\gamma(x)_{1:n})\\
        &= \frac{1}{2} \langle{\gamma(x)_{1:n}}|V\rho V^\dagger |{\gamma(x)_{1:n}}\rangle,
    \end{align}
    which is identical to $P(x)$.

    \textit{Case 3}: All $g_i$ act as $I$ or $Y$ on the ancilla. The proof for this case is identical to that of \textit{Case 2} by replacing $X$ with $Y$.

    \textit{Case 4}: There exist at least two $g_i$'s acting non-trivially on the ancilla with different Pauli operators. In this case, suppose $g_{n}$ acts $Z$ on the ancilla and $g_{n+1}$ acts $X$ on the ancilla. Then, we can multiply $(-1)^{x_n}g_{n}$ and/or $(-1)^{x_{n+1}}g_{n+1}$ to other generators that act nontrivially on the ancilla so that $g_n$ and $g_{n+1}$ are the only generators that act nontrivially on the ancilla. Specifically, for some one-to-one and onto function $\alpha: \{0,1\}^{n+1} \rightarrow \{0,1\}^{n+1}$, we have
    \begin{equation}
        S = \langle (-1)^{\alpha(x)_1}g'_1, \dots, (-1)^{\alpha(x)_{n-1}}g'_{n-1}, (-1)^{\alpha(x)_n}g_n, (-1)^{\alpha(x)_{n+1}}g_{n+1} \rangle,
    \end{equation}
    where $g'_1,\dots, g_{n-1}'$ are in the forms of $g'_{j}=h_{j}\otimes I$ for $j=1,\dots ,n-1$ for $h_1,\dots,h_{n-1} \in \mathcal{P}_n$. We also denote $g_n=h_n \otimes Z$ and $g_{n+1}=h_{n+1} \otimes X$. Here, it is apparent that $h_1,\dots,h_{n-1}$ are independent and mutually commuting Pauli operators in $\mathcal{P}_n$.
    
    One can see that this is true even when appending $h_n$, i.e., $h_1,\dots,h_{n}$ are independent and mutually commuting. First, it is clear that $h_n$ commutes with all $h_1,\dots,h_{n-1}$, as $g_n$ commutes with $g_1,\dots,g_{n-1}$. Second, if $h_1,\dots,h_{n}$ is not independent, $h_n \in \langle h_1,\dots,h_{n-1}\rangle$. Therefore, there exists $i_1, \dots,i_m \le n-1$ such that $h_{n}=h_{i_1}h_{i_2}\dots h_{i_m}$. Then, we have $g\in S$ such that $g=g_{i_1}\dots g_{i_m} g_n = I\otimes Z$. It contradicts with the fact that $S$ is a abelian subgroup of $P_{n+1}$, as $g$ anticommutes with $g_{n+1}$. Therefore, $h_1,\dots,h_{n}$ are independent and mutually commuting.
    
    Let $V \in \mathcal{C}_n$ be a $n$-qubit Clifford gate such that $V^{\dagger} Z_i V = h_i$ for all $i=1,\dots,n$. Then,
    \begin{align}
        U^{\dagger}\ketbra{x}{x}U
        &= \left(\prod_{j=1}^{n-1}\frac{I_{SA}+(-1)^{\alpha(x)_j}g'_{j}}{2}\right)\frac{I_{SA}+(-1)^{\alpha(x)_n}g_{n}}{2} \cdot \frac{I_{SA}+(-1)^{\alpha(x)_{n+1}}g_{n+1}}{2}\\
        &= \left(\prod_{j=1}^{n-1}\frac{I_S+(-1)^{\alpha(x)_j}h_j}{2} \otimes I_A\right)
        \frac{I_{SA}+(-1)^{\alpha(x)_n}h_n\otimes Z}{2}
        \cdot
        \frac{I_{SA}+(-1)^{\alpha(x)_{n+1}}h_{n+1}\otimes X}{2}.
    \end{align}
    Now, note that
    \begin{align}
        (I_{S}\otimes \bra{0})\frac{I_{SA}+(-1)^{\alpha(x)_n}h_n\otimes Z}{2}
        &= \frac{I_{S}+(-1)^{\alpha(x)_{n}} h_n}{2}(I_{S}\otimes \bra{0})
    \end{align}
    and
    \begin{align}
        (I_{S}\otimes \bra{0})\frac{I_{SA}+(-1)^{\alpha(x)_{n+1}}h_{n+1}\otimes X}{2}(I_{S}\otimes \ket{0})
        &= \frac{I_S}{2}.
    \end{align}
    Therefore,
    \begin{align}
        (I_{S}\otimes \bra{0})U^{\dagger}\ketbra{x}{x}U (I_{S}\otimes \ket{0})
        &= \left(\prod_{j=1}^{n-1}\frac{I_S+(-1)^{\alpha(x)_j}h_j}{2}\right)
        \frac{I_S+(-1)^{\alpha(x)_n}h_n}{2}\cdot \frac{1}{2}\\
        &= V^\dagger\left(\prod_{i=1}^{n}\frac{I_S+(-1)^{\alpha(x)_i}Z_i}{2}\right)V\cdot \frac{1}{2}\\
        &= \frac{1}{2}V^{\dagger}\ketbra{\alpha(x)_{1:n}}{\alpha(x)_{1:n}}V.
    \end{align}

    We can obtain the same output distribution without the ancilla with the circuit that applies $V$ on the input of $\rho$ and then measures every qubit. Let the output distribution of this circuit be $Q$:
    \begin{equation}
        Q(y) = \langle{y}|V\rho V^{\dagger}|{y}\rangle,
    \end{equation}
    for all $y\in\{0,1\}^n$. Once we get an outcome $y$, we process it as follows:
    \begin{equation}
        y \mapsto
        x = 
        \begin{cases}
            \alpha^{-1}(y0), \quad &\text{with probability 1/2}\\
            \alpha^{-1}(y1), &\text{with probability 1/2}
        \end{cases}
    \end{equation}
    Let $Q'$ denote the distribution of $x$ after this post-processing. Then,
    \begin{align}
        Q'(x) &= \frac{1}{2}Q(\alpha(x)_{1:n})\\
        &= \frac{1}{2} \langle{\alpha(x)_{1:n}}|V\rho V^\dagger |{\alpha(x)_{1:n}}\rangle,
    \end{align}
    which is identical to $P(x)$.
\end{proof}

\end{document}